\documentclass[twoside,twocolumn]{article}

\usepackage[letterpaper,hcentering,textwidth=6.75in,textheight=9.25in,top=1in,headheight=12pt,headsep=15pt]{geometry}
\usepackage{fancyhdr}
\renewcommand{\headrulewidth}{0.5pt}
\renewenvironment{abstract}
{\centerline{\large\bfseries Abstract}\vspace{0.5ex}\begin{quote}}
{\par\end{quote}\vskip 1ex}
\usepackage[round]{natbib}

\usepackage{amssymb}
\usepackage{amsmath}
\usepackage{enumitem}
\usepackage{bm}
\usepackage{xcolor}
\usepackage[hidelinks]{hyperref}
\usepackage{algorithm}
\usepackage{algpseudocode}
\algtext*{EndIf} 
\usepackage{graphicx}
\usepackage{booktabs}

\newcommand{\lt}{\left}
\newcommand{\rt}{\right}
\newcommand{\indicator}[1]{\mathbf{1}{\left\{ #1 \right\}}}
\newcommand{\E}{\mathbb{E}}
\newcommand{\V}{\mathrm{Var}}
\newcommand{\R}{\mathbb{R}}

\newcommand{\indep}{\perp\!\!\!\perp}

\newcommand{\argmin}{\operatorname*{arg\,min}}
\newcommand{\argmax}{\operatorname*{arg\,max}}

\DeclareMathOperator{\Cov}{Cov}
\newcommand{\cA}{\mathcal{A}}

\newcommand{\cD}{\mathcal{D}}

\newcommand{\cH}{\mathcal{H}}

\newcommand{\cN}{\mathcal{N}}
\newcommand{\cP}{\mathcal{P}}

\newcommand{\cX}{\mathcal{X}}

\newcommand{\cZ}{\mathcal{Z}}

\newcommand{\TV}{\operatorname{d_{\mathrm{TV}}}}

\newcommand{\sA}{\mathsf{A}}

\newcommand{\sE}{\mathsf{E}}
\newcommand{\sK}{\mathsf{K}}

\newtheorem{assumption}{Assumption}

\newtheorem{definition}{Definition}

\newtheorem{proposition}{Proposition}
\newtheorem{theorem}{Theorem}

\newtheorem{lemma}{Lemma}
\newenvironment{proof}{\paragraph{Proof:}}{\hfill$\square$}

\begin{document}

\title{\textbf{Policy Learning with Weak Signals}}

\author{
{\large Benedikt Koch}\\
{\small Harvard University}\\
{\footnotesize benedikt\_koch@g.harvard.edu}
\and
{\large Winston Chou}\\
{\small Netflix}\\
{\footnotesize wchou@netflix.com}
\and
{\large Aur\'elien Bibaut}\\
{\small Netflix}\\
{\footnotesize abibaut@netflix.com}
\and
{\large Nathan Kallus}\\
{\small Netflix}\\
{\footnotesize nkallus@netflix.com}
}
\date{}
\maketitle
\thispagestyle{empty}

\begin{abstract}
Policy learning in digital experimentation faces three challenges: weak signal-to-noise ratios, rich covariate spaces, and massive data volumes. We formalize this regime by modeling treatment-effect estimates from increasingly fine covariate partitions as Gaussian observations with bounded signal-to-noise ratios. We establish that, in general, the optimal treatment policy is not learnable in this setting. Even learning the optimal policy value suffers from impractically slow rates. However, when treatment effects vary smoothly, we derive minimax-adaptive policies based on linear smoothers that achieve vanishing welfare regret. We demonstrate the practical value of our framework by applying it to large-scale real-world experiments at Netflix, showing that personalized linear-smoothing policies can dominate unpersonalized policies even in this challenging empirical setting.
\end{abstract}

\section{Introduction}

Randomized experiments are increasingly used not only to decide whether to deploy a treatment to an entire population, but also to learn a personalized treatment policy
\citep{manski2004statistical, hirano2009asymptotics, kitagawa2018should, athey2021policy}. This is of special interest to digital technology firms which strive to deliver a personalized product experience to each individual user \citep{gomez2015netflix, tu2021personalized}.

However, learning effective treatment policies in digital experiments faces three challenges. First, because digital firms run many experiments, large step-change improvements are rare, so that treatment effects tend to be small relative to measurement error even with millions of users \citep{peysakhovich2018learning, azevedo2020b, bibaut2024learning}. Second, the covariate space is often very rich, which helps detect treatment effect heterogeneity, but increases the the number of treatment decisions to be learned. Lastly, data volume and privacy constraints often only allow storage and analysis of summary- rather than unit-level records \citep{chou2021randomizedcontrolledtrialsdata, wong2019efficientcomputationlinearmodel}.

We formalize this regime by modeling treatment-effect estimates as a Gaussian sequence model over an increasing number of covariate cells, whose signal-to-noise ratios remain bounded. We show that, in such a regime, learning the optimal \emph{personalized} policy is asymptotically impossible: no decision rule yields vanishing worst-case and (least-favorable prior) Bayes regret, extending the scalar results of \citet{hirano2009asymptotics} and \citet{tetenov2012statistical}. Adapting results from \citet{cai2011}, \citet{lepski1999estimation}, and \citet{han2020estimation}, we show that the optimal policy \emph{value} is estimable, although its minimax-optimal logarithmic rate is too slow to be useful in practice.

Nevertheless, when the conditional average treatment effect (CATE) function varies smoothly (is $\beta$-H\"older) across $K$ covariate cells, simple linear smoothers such as trees, kernels, or spline regression can borrow signal across neighboring cells and induce policies with vanishing welfare-regret. We derive sharp regret bounds for linear smoothers that trade off approximation and estimation error. Under a quasi-uniform design, the right amount of signal sharing between cells balances the two and yields the rate $K^{-\beta/(2\beta+d)}$ where $d$ is the covariate dimension, matching classical rates for nonparametric regression and classification \citep{stone1982optimal, yang1999minimax, audibert2007fast}. Our theoretical contribution is to establish that tree, kernel, and spline policies adapt to unknown smoothness and attain this rate under nonuniform welfare weights and a fixed-design, heteroskedastic Gaussian observation model. We obtain a matching minimax lower bound by adapting Assouad's lemma to weighted welfare and randomized policies.

We demonstrate the practical relevance of our methodological framework by applying it to two large-scale real-world digital experiments at Netflix. We show that, in this challenging empirical setting, treatment effect policies based on linear smoothing estimators can achieve robust welfare gains over unpersonalized policies.  Our applications underscore both the importance of personalization and the empirical value of linear smoothing for learning useful treatment policies.

\section{Framework and Notation}
\label{sec:setup}

We begin by introducing our notational framework for policy learning with weak signals.

\subsection{Tasks, Policies, and Welfare}
Consider a collection of local decision problems, indexed by $k=1,\ldots,K$, which we call \emph{tasks}. Each task involves a binary treatment choice, has an associated \emph{treatment effect} $\tau_k\in\R$, and a known, deterministic \emph{task feature} $x_k\in\cX$. Each task is also associated with a known \emph{task weight} $p_k > 0$ with $\sum_{k=1}^K p_k = 1$.

In our policy learning setting, a task is an element of a discrete covariate space $\cX=\{x_1,\ldots,x_K\}$, with $\tau_k = \E[Y(1) - Y(0) \mid X = x_k]$ the conditional average treatment effect at value $x_k$, and $p_k = P(X = x_k)$. The task-level decision problem is whether to treat units having that covariate value.\footnote{Continuous covariates are accommodated by partitioning $\cX=\bigcup_{k=1}^K \cX_k$ and setting $\tau_k=\E[\tau(X)\mid X\in\cX_k]$, where $\tau(x)=\E[Y(1)-Y(0)\mid X=x]$.}

A \emph{policy} $\pi = (\pi_1, \ldots, \pi_K) \in [0, 1]^K$ assigns treatment probability $\pi_k$ to task $k$ and yields the \emph{welfare} given the $K$ tasks
\[
    V_K(\pi) = \sum_{k=1}^K p_k \tau_k \pi_k,
\]
where we drop the policy independent baseline. For shortness, we write $V(\pi)=V_K(\pi)$. The \emph{oracle policy} $\pi_k^\ast = \indicator{\tau_k > 0}$ maximizes $V$ and induces the \emph{optimal value}
\begin{align}\label{eq:optimal_value}
    V^\ast = \sup_{\pi \in [0, 1]^K} V(\pi) = \sum_{k=1}^K p_k \tau_k^+,
\end{align}
where $\tau^+ = \max\{\tau, 0\}$. The oracle policy may be non-unique when $\tau_k = 0$ for some tasks, but the optimal value $V^\ast$ is unique \citep{luedtke2016statistical, shi2020breaking, whitehouse2025inference}. Recovering $V^\ast$, and a policy that attains it, is our goal.

\subsection{Observation Model and Decision Rules}
The decision maker uses task-level data $(\widehat\tau_k, \widehat \sigma_k, x_k,n_k)_{k=1}^K$, where $\widehat\tau_k$ is an estimate of $\tau_k$ with standard error $\sigma_k/\sqrt{n_k}$, based on $n_k$ units. For example, $\widehat\tau_k$ may be the difference-in-means estimator in a randomized trial. A \emph{statistical decision rule} is a measurable map $\widetilde\pi: \R^K \times [0,\infty)^K \times \cX^K \to [0, 1]^K$, $(\widehat\tau, \widehat \sigma, x) \mapsto \widetilde\pi(\widehat\tau, \widehat\sigma, x)$. Each coordinate may depend on the entire triplet, so rules that pool or shrink across tasks (e.g., empirical Bayes) are included. Since $x$ is fixed throughout, we suppress its input.

Even if the decision maker observes unit-level data, task-level summaries can sufficiently capture treatment information, in large samples under suitable conditions, see Appendix~\ref{app:limit_experiments}. Such summaries are also common in industrial settings, where privacy concerns, data warehousing costs, and scale limit the storage and processing of unit-level records.  

Because the welfare $V(\widetilde\pi(\widehat\tau, \widehat \sigma))$ of a decision rule is random, we assess rules by their \emph{expected welfare} $W(\widetilde\pi) = \E[V(\widetilde\pi(\widehat\tau, \widehat \sigma))]$, with expectation over the randomness of $(\widehat\tau, \widehat \sigma)$, and aim to minimize the \emph{regret}
\[
    V^\ast - W(\widetilde\pi) = \sum_{k = 1}^K p_k \tau_k \lt(\indicator{\tau_k > 0} - \E[\widetilde\pi_k(\widehat \tau, \widehat \sigma)]\rt) \geq 0.
\]

We assume that the task-level estimates follow a Gaussian sequence model centered at the true effects \citep{peysakhovich2018learning, azevedo2020b, sudijono2024optimizing, tripuraneni2024choosing, chen2025compound,chen2026normal}, with justification in Appendix~\ref{app:limit_experiments}.

\begin{assumption}[Gaussian Sequence Model]\label{ass:gaussian_model}
For $k = 1, \ldots, K$,
\[
    \widehat\tau_k
    \overset{\mathrm{ind}}{\sim}
    \mathcal N\lt(\tau_k,\frac{\sigma_k^2}{n_k}\rt),
\]
where the standard deviations $\sigma_k$ are bounded, i.e., $\underline{\sigma} \leq \sigma_k \leq \overline{\sigma}$ with known constants $\underline{\sigma}, \overline{\sigma} > 0$.
\end{assumption}

For expositional simplicity, we set $n_k=n$, but our results extend (up to constants) to settings where the ratios $n_k/n$ are bounded away from zero and infinity by replacing $\sigma_k$ with $\tilde\sigma_k = \sigma_k \sqrt{n / n_k}$ below.

\subsection{The Weak-Signal Regime}

The defining feature of our setting is that treatment effects are small relative to sampling noise. Under this setting, even interventions worth deploying move the outcome of interest by an amount comparable to its standard error \citep{hirano2009asymptotics,mikusheva2022inference, azevedo2020b, chou2025evaluating}. Assumption~\ref{ass:weak_signal} encodes this by fixing the information per task.

\begin{assumption}[Weak Signal]\label{ass:weak_signal}
There is a constant $H > 0$ such that $\tau_k = h_k / \sqrt{n}$ with $\sup_k |h_k| \leq H$ for all $k$.
\end{assumption}

Under Assumptions~\ref{ass:gaussian_model}--\ref{ass:weak_signal}, the observation model can equivalently be written in terms of the rescaled statistics $Z_k := \sqrt{n} \widehat\tau_k$,
\begin{equation}\label{eq:weak_model}
    Z_k \overset{\mathrm{ind}}{\sim} \cN(h_k, \sigma_k^2), \qquad k = 1, \dots, K.
\end{equation}
Similarly, we equivalently write decision rules as $\widetilde\pi(Z,\widehat\sigma) = \widetilde\pi (\widehat\tau, \widehat \sigma)$. Thus, each task has bounded signal-to-noise ratio $h_k / \sigma_k = O(1)$. To return to the original scale, we target the rescaled oracle value $\sqrt{n} V^\ast = \sum_{k=1}^K p_k h_k^+$.

We study the regime in which the number of tasks grows, so that the covariate space becomes richer, $K\to\infty$, with approximately balanced task sizes $n_k \asymp n$, where $n$ is large but fixed.\footnote{$a_K\asymp b_K$ means that there exist constants $0<c\leq C<\infty$ such that $cb_K\leq a_K\leq Cb_K$ for all sufficiently large $K$.} Thus the total sample size $N=\sum_{k=1}^K n_k\asymp Kn$ grows as additional observations support increasingly fine partitions of $\cX$, but per-task signal-to-noise ratios remain bounded under Assumption~\ref{ass:weak_signal}. Our empirical analysis in Section~\ref{sec:netflix} is consistent with this regime.

\section{Policy and Value Learning}\label{sec:policy_learning}

We characterize what can be learned in this weak-signal regime. We show that, in the general case without any structure linking tasks, no statistical decision rule can uniformly recover the oracle policy with vanishing worst-case regret. Nevertheless, the scalar optimal value $\sqrt{n} V^\ast$ can still be recovered by polynomial approximation of the nonsmooth functional $h_k^+$ if the standard errors are known. However, because the resulting estimator converges very slowly, learning this value is impractical even if feasible.

Let $\cP_K$ denote all joint laws of $(Z, \widehat \sigma)$ satisfying Assumptions~\ref{ass:gaussian_model} and~\ref{ass:weak_signal}. Write $\bm h = (h_1, \ldots, h_K)$ and $\bm \sigma = (\sigma_1, \ldots, \sigma_K)$. Then, an element $P_{\bm h, \bm \sigma} \in \cP_K$ is a joint law of $(Z, \widehat \sigma)$ with Gaussian marginal $Z \sim \cN(\bm h, \operatorname{diag}(\sigma_1^2, \ldots, \sigma_K^2))$. Thus, laws $P_{\bm h, \bm \sigma}, P_{\bm h, \bm \sigma}^\prime \in \cP_K$ share the same $Z$-marginal but may differ in the conditional law of $\widehat \sigma$ given $Z$.
Define the optimal value, expected welfare and scaled regret
\begin{align*}
    V^\ast_{\bm h} &= \frac{1}{\sqrt{n}} \sum_{k=1}^K p_k h_k^+,
    \\
    W_{P_{\bm h, \bm \sigma}}(\widetilde\pi) &= \frac{1}{\sqrt{n}} \sum_{k=1}^K p_k h_k \E_{P_{\bm h, \bm \sigma}}[\widetilde\pi_k(Z, \widehat \sigma)],
    \\
    R_K(P_{\bm h, \bm \sigma}, \widetilde\pi) &:= \sqrt{n} \{V^\ast_{\bm h} - W_{P_{\bm h, \bm \sigma}}(\widetilde\pi)\},
\end{align*}
We write $\E_{\bm h,\bm\sigma}$ instead of $\E_{P_{\bm h, \bm \sigma}}$ when the expectation only involves $Z$.

\subsection{An Optimal Policy is not Learnable}

Theorem~\ref{thm:minimax_regret} shows that learning an optimal policy incurs nonvanishing worst-case welfare regret.
\begin{theorem}[Unattainable Optimal Policy]\label{thm:minimax_regret}
   Suppose Assumptions~\ref{ass:gaussian_model} and~\ref{ass:weak_signal} hold. Then,
   \begin{align*}
        &\inf_{\widetilde\pi} \sup_{P_{\bm h, \bm \sigma} \in \cP_K} R_K(P_{\bm h, \bm \sigma}, \widetilde\pi) 
        \\
        &=  \max_{\Pi \in \mathfrak P} \inf_{\widetilde\pi} \int R_K(P_{\bm h, \bm \sigma},\widetilde\pi)
\Pi(dP_{\bm h, \bm \sigma})
        \\
        &=\max_{0 \leq t \leq H} t\Phi \lt(-\frac{t}{\overline \sigma}\rt),
    \end{align*}
   where the infima range over rules $\widetilde\pi: \R^K \times [0, \infty)^K \to [0, 1]^K$ and $\mathfrak P$ is the set of Borel priors on $\cP_K$.
    The plug-in rule $\widetilde\pi_k^*(Z, \widehat \sigma)=\indicator{Z_k>0}$ is both minimax and Bayes under the least-favorable prior.
\end{theorem}
The proof is given in Appendix~\ref{app:minimax_regret}. Since $n$ does not grow and the minimax regret value in Theorem~\ref{thm:minimax_regret} does not depend on $K$, the regret need not vanish as $K\to\infty$. Thus, in the weak-signal regime, no decision rule achieves vanishing worst-case regret over any model class that contains the independent Gaussian submodel with bounded noise and uninformative task features. Furthermore, the minimax optimal rule ignores $\widehat\sigma$ and is separable, i.e., its decision for task $k$ depends only on $Z_k$.

If $K = 1$, $\bm\sigma$ is known, and $H$ is large, Theorem~\ref{thm:minimax_regret} matches the
exact minimax regret in the scalar Gaussian model \citep{tetenov2012statistical}, which is also
the local asymptotic minimax bound in \citep[Theorem~3.5]{hirano2009asymptotics}. Our proof is based on a weighted Assouad argument (Lemma~\ref{lem:weighted_assouad}) rather than a complete-class theorem.

\subsection{The Optimal Value is Learnable}
The value of the optimal policy is learnable if we assume that the standard deviations are known. However, we show that its estimation error converges very slowly in $K$.  Writing $h_k^+ = \frac{1}{2}(h_k + |h_k|)$,
\[
    \sqrt{n} V_{\bm h}^\ast = \frac{1}{2} \sum_{k=1}^K p_k h_k + \frac{1}{2} \sum_{k=1}^K p_k |h_k|.
\]
The first term is a smooth linear functional in $\bm h$ and can be estimated at parametric rate by $\frac{1}{2} \sum_{k=1}^K p_k Z_k$. The difficulty lies in the second term as $| \cdot |$ is nonsmooth at $h_k = 0$, so averaging $|Z_k|$ leaves nonvanishing bias. Following \citet{cai2011,han2020estimation}, we approximate $|h_k|$ by a polynomial, which we estimate in turn unbiasedly using Hermite polynomials using the known standard errors. The polynomial degree controls the tradeoff between the bias from approximating $|h_k|$ and the variance of the estimator. 

Let $Q^\ast_m(x) = \sum_{j = 0}^m q^\ast_{j, m} x^{2j}$ be the best uniform polynomial approximation of $|x|$ on $[-H, H]$ of degree $2m$ (necessarily even). Denote with $\Theta_j(x) = (-1)^j e^{\frac{x^2}{2}} \frac{\partial^j}{\partial x^j}(e^{-\frac{x^2}{2}})$ the probabilist's Hermite polynomial of degree $j$. A key feature of $\Theta_j$ is that for $Z_k \sim \cN(h_k, \sigma_k^2)$, we have $\E[\sigma_k^{2j} \Theta_{2j}(Z_k/\sigma_k)] = h_k^{2j}$. This motivates the estimator
\[
    \widehat V^*_m:= \frac{1}{2 \sqrt{n}} \sum_{k=1}^K p_k \lt(Z_k + \sum_{j = 0}^m q^\ast_{j, m} \sigma_k^{2j} \Theta_{2j} \lt(\frac{Z_k}{\sigma_k}\rt) \rt).
\]

The following result is an adaptation of \citet[Theorem~3 and Theorem~4]{cai2011}.
\begin{theorem}[Estimating the Optimal Value]\label{thm:value}
Suppose Assumptions~\ref{ass:gaussian_model}--\ref{ass:weak_signal} hold and that the true standard errors $\bm \sigma$ are known. Suppose further that $\max_k p_k = O(1/K)$ and that $m_K = \lt \lfloor \frac{1}{4} \log(K) / (\log(\log(K))) \rt \rfloor$. Then,
\begin{align*}
    \sup_{\bm h \in [-H, H]^K} &\E_{\bm h, \bm \sigma}\lt[\lt(\sqrt{n} \{\widehat V^\ast_{m_K} - V_{\bm h}^\ast\}\rt)^2\rt]
    \\
    &\qquad \qquad \qquad \lesssim_{H, \overline{\sigma}, \underline{\sigma}} \lt(\frac{\log(\log (K))}{\log (K)}\rt)^2,
\end{align*}
This rate is minimax optimal in the submodel $p_k = 1/K, \sigma_k = \sigma$ for all $k = 1, \ldots, K$ for some $\sigma > 0$. \footnote{$a\lesssim b$ means $a\leq Cb$ for a universal constant $C$.  Subscripts on $\lesssim$ list the quantities on which $C$ may depend.}
\end{theorem}

The proof is given in Appendix~\ref{app:value}. As with
Theorem~\ref{thm:minimax_regret}, the optimality claim covers estimators that
additionally use the task features. The results in this section show that, in the weak-signal regime, the sign vector $\{\indicator{\tau_k>0}\}_{k=1}^K$ cannot be 
uniformly recovered, but its associated optimal value
$\sum_k p_k \tau_k^+$ can be estimated consistently when the number of tasks grows, even when per-task information is fixed or decreases. However, the slow logarithmic rate limits the utility of the resulting optimal estimator in practice. 

\section{Policy Learning by Linear Smoothing}\label{sec:linear_smoothing}
In Section~\ref{sec:policy_learning}, we showed that, without restrictions linking tasks, no decision rule can achieve vanishing worst-case welfare regret. Writing $h_k=g(x_k)$, this setting is equivalent to allowing $g$ to vary arbitrarily over the task features. We now show that if $g$ is smooth, one can borrow information across tasks, making treatment decisions easier while introducing an approximation bias that can be controlled.

\begin{assumption}[H\"older-Smooth Signal]
\label{ass:holder_smoothness}
For fixed $\beta\in(0,1]$ and $L,H>0$, the task-level signals satisfy
$h_k=g(x_k)$ for every $k\in[K]$, where
$g\in\cH_d(\beta,L,H)$ and
\begin{align*}
    \cH_d(\beta, L, H) &:= \Big\{g: \R^d \to \R: \|g\|_\infty \leq H, 
    \\
    &\quad |g(x) - g(y)| \leq L \|x - y\|_2^\beta, \, \forall x, y \in \R^d\Big\}.
\end{align*}
\end{assumption}
Assumption~\ref{ass:holder_smoothness} is standard in nonparametric theory \citep{stone1982optimal,lepski1999estimation,audibert2007fast}. We restrict attention to $\beta\in(0,1]$, which does not require differentiability of the CATE.

For $g\in\cH_d(\beta,L,H)$, write $P_{g,\bm\sigma}$ for a joint law in $\cP_K$ with $h_k=g(x_k)$. Define the oracle value, expected welfare and scaled regret by, respectively,
\begin{align*}
    & V^\ast_g := \frac{1}{\sqrt{n}}\sum_{k = 1}^K p_k g(x_k)^+ \\
    & W_{P_{g, \bm{\sigma}}}(\widetilde\pi) := \E_{P_{g, \bm{\sigma}}}[V(\widetilde\pi(Z, \widehat \sigma))] \\
    & R_K(P_{g, \bm{\sigma}}, \widetilde\pi) := \sqrt{n}\lt\{V^\ast_g - W_{P_{g, \bm{\sigma}}}(\widetilde\pi)\rt\}.
\end{align*}
Appendix~\ref{app:binary_classification} connects our treatment-choice problem to an auxiliary binary classification and shows that $R_K$ is equivalent to the corresponding expected excess risk.

\subsection{Linear Smoothing}
We establish that linear smoothers achieve vanishing regret and are minimax optimal in our regime.

\begin{definition}[Linear Smoother]
\label{def:linear_smoother}
Let $w_{j} :\R^d \to \R$ be deterministic weight functions
satisfying $\sum_{j=1}^K w_{j}(x)=1.$
Then
\[
    \widehat g_w(x)
    :=
    \sum_{j=1}^K w_{j}(x)Z_j,
\]
is a linear smoother for signal $g$ and induces the (threshold) policy $\widetilde\pi_w(x) = \indicator{\widehat g_w(x)>0}$.
\end{definition}
Definition~\ref{def:linear_smoother} permits negative weights and ignores $\widehat \sigma$. We focus on linear smoothers because they are interpretable and are easily fitted to task-level data, making them desirable in practice. Moreover, for estimation, linear estimators have been shown to be minimax or near minimax optimal in certain settings \citep{donoho1990minimax,donoho1994minimax, donoho1998minimax}.
Under Assumption~\ref{ass:holder_smoothness}, we show that rules based on linear smoothers can yield vanishing welfare regret and achieve minimax-optimal rates.\footnote{In Proposition~\ref{pro:welfare_decomposition} (Appendix~\ref{app:additional_theoretical_results}), we show that the expected welfare of a policy induced by a linear smoother, $W_{P_{g, \bm{\sigma}}}(\widetilde\pi_w)$, decomposes into a task-invariant baseline term and a personalization term that captures whether tasks with larger treatment effects are more likely to receive treatment.}

We focus on three classes of linear smoothers: partition pooling ($\mathrm{P}$), kernel smoothing ($\mathrm{K}$) and series regression ($\mathrm{S}$), which are summarized in
Table~\ref{tab:smoothers}. Each is indexed by a
\emph{resolution} parameter $r>0$, which controls the scale of smoothing. Precise definitions of these classes are in
Appendix~\ref{app:smoothers}.

\begin{table*}[ht]
\centering
\small
\renewcommand{\arraystretch}{1.6}
\begin{tabular}{@{}p{0.15\linewidth} p{0.2\linewidth} p{0.6\linewidth}@{}}
\toprule
Smoother & Weight $w_{r,j}(x)$ & Conditions \\
\midrule
Partition pooling
  & $\dfrac{p_j\,\indicator{x_j\in A_r(x)}}{\sum_{\ell:\,x_\ell\in A_r(x)}p_\ell}$
  & $\cA_r$ a finite partition of $\cX$; $A_r(x) \in \cA_r$ the cell containing $x$; \newline 
     $\max_{A\in\cA_r}\sup_{x,y\in A}\|x-y\|_2\le r$; \\
Kernel smoothing
  & $\dfrac{\sK\lt((x_j-x)/r\rt)}{\sum_{\ell=1}^K\sK\lt((x_\ell-x)/r\rt)}$
  & $\sK:\R^d\to[0,\infty)$ bounded, supported on $B(0,1)$,
    continuous at $0$, \newline $\sK(0)>0$. \\
Series regression
  & $a_j\, b_r(x)^\top G_{r,\bm a}^{-1} b_r(x_j)$
  & $a_k>0$, $\sum_k a_k=1$; $b_r:\cX\to\R^{q_r}$ with $q_r\le K$,
    $q_r\asymp r^{-d}$; \newline constants $\in\operatorname{span}(b_r)$;
    $G_{r,\bm a}:=\sum_{k}a_k b_r(x_k)b_r(x_k)^\top$ nonsingular. \\
\bottomrule
\end{tabular}
\caption{Linear smoothers at resolution $r$.}
\label{tab:smoothers}
\end{table*}

\subsection{Approximation--Estimation Tradeoff}

Linear smoothers can reduce variance at the cost of introducing bias. For a linear smoother $\widehat g_w$, denote its mean $m_w(x) := \E_{g, \bm{\sigma}}[\widehat g_w(x)]  = \sum_{k=1}^K w_k(x)g(x_k)$ and variance $s_w^2(x) := \text{Var}_{g, \bm{\sigma}}(\widehat g_w(x)) = \sum_{k=1}^K w_k(x)^2\sigma_k^2$.

\begin{proposition}[Linear-Smoothing Tradeoff]\label{pro:smoothing_tradeoff}
Suppose Assumptions~\ref{ass:gaussian_model}--\ref{ass:holder_smoothness} hold. The policy induced by
any linear smoother $\widehat g_w$ satisfies
    \begin{align*}
         R_K(P_{g, \bm{\sigma}}, \widetilde\pi_w) \leq \min\{H, \sA_g(w) + \gamma^\ast \sE_{\bm{\sigma}}(w)\},
    \end{align*}
    where $\sA_g(w) = \sum_{k = 1}^K p_k |g(x_k) - m_w(x_k)|$ is the \emph{approximation} and $\sE_{\bm{\sigma}}(w) = \sum_{k = 1}^K p_k s_w(x_k)$ is the \emph{estimation} error, and $\gamma^\ast =\sup_{t\geq0}t\Phi(-t)$. This bound is sharp over the class of linear smoothers.
\end{proposition}
The proof is given in Appendix~\ref{app:smoothing_tradeoff}. The two terms in Proposition~\ref{pro:smoothing_tradeoff} represent a bias--variance tradeoff and formalize the costs and benefits of smoothing. As the weights become more concentrated around a task's covariate, the smoother estimates the signal more locally and reduces approximation error, but combines fewer noisy task estimates and therefore generally increases estimation error. More diffuse weights, on the contrary, yield better estimation accuracy at the cost of larger approximation error. The resolution parameter $r$ governs this tradeoff.

Under Assumption~\ref{ass:holder_smoothness}, the smoothers specified in Table~\ref{tab:smoothers} satisfy
$\sA_g(w_r^T)\lesssim Lr^\beta$ for $T\in\{\mathrm P,\mathrm K,\mathrm S\}$. To control the corresponding estimation errors, we impose the following condition on the design, which is a discrete analogue of the regularity condition of \citet{david1991singular} and implied by the strong density assumption of \citet{audibert2007fast}.

Let $Q_K = \sum_{k = 1}^K p_k \delta_{x_k}$ denote the task measure and let $B(x, r) = \{y \in \R^d: \|x - y\|_2 \leq r\}$ be the closed Euclidean ball centered at $x$ with radius $r$.
\begin{assumption}[Quasi-Uniform Design]
\label{ass:quasi_uniform_design}
There exist constants $0<c\leq C<\infty$, $a_0, r_0>0$, such that
\[
    c r^d
    \leq
    Q_K\lt(B(x_k,r)\rt)
    \leq
    Cr^d
\]
for every $k\in[K]$ and $r\in[a_0 K^{-1/d},r_0]$.
\end{assumption}
Assumption~\ref{ass:quasi_uniform_design} requires the welfare mass of each ball centered at a task to be proportional to its volume. The cutoff at $K^{-1/d}$ is necessary because $Q_K$ is discrete. Indeed, for arbitrarily small $r$, we have $Q_K(B(x_k,r))\geq p_k>0$, so the mass cannot shrink to zero. The resolution $K^{-1/d}$ corresponds to the regular spacing of $K$ points in $d$ dimensions.

\begin{theorem}[Regret Bounds]\label{thm:regret_linear_smoother}
    Suppose Assumptions~\ref{ass:gaussian_model}--\ref{ass:quasi_uniform_design}
    hold, and let $r\in[a_0K^{-1/d},r_0]$.
    Let $\widetilde\pi_{w_r^T}$, $T\in\{\mathrm P,\mathrm K,\mathrm S\}$,
    be a linear smoothing policy satisfying the conditions in
    Appendix~\ref{app:smoother_conditions}. Then there exists $C > 0$ independent of $(K, r, g, \overline{\sigma})$, such that 
    \begin{align*}
        R_K(P_{g,\bm{\sigma}},\widetilde\pi_{w_r^T}) \leq C \lt\{Lr^\beta + \frac{\overline{\sigma}}{\sqrt{K r^d}}\rt\}.
    \end{align*}
     In particular, for $r_K^\ast = \lt( \frac{d \overline{\sigma}}{2\beta L\sqrt K}\rt)^{\frac{2}{2\beta+d}}$,
    \begin{align*}
        R_K(P_{g,\bm{\sigma}},\widetilde\pi_{w_{r_K^\ast}^T}) \leq C L^{d/(2\beta+d)}
    \overline\sigma^{2\beta/(2\beta+d)}
    K^{-\beta/(2\beta+d)}.
    \end{align*}
\end{theorem}
The proof is given in Appendix~\ref{app:regret_linear_smoother}. Gaussianity (Assumption~\ref{ass:gaussian_model}) is not necessary for the result; any distribution with appropriate second moment bound suffices. Without a margin condition, the rate $K^{-\beta/(2\beta+d)}$ matches the classical nonparametric regression rate in \citep{stone1982optimal} and the binary classification rate in \citet{yang1999minimax} and \citet{audibert2007fast}. Novel to the existing literature, Theorem~\ref{thm:regret_linear_smoother} derives this rate for three concrete smoothers for welfare regret under a fixed-design, heteroskedastic Gaussian observation model with nonuniform weights.

Theorem~\ref{thm:regret_linear_smoother} explicitly shows that
\[
    \sA_g(w_r)\lesssim Lr^\beta,
    \qquad
    \sE_{\bm\sigma}(w_r)
    \lesssim
    \frac{\overline\sigma}{\sqrt{Kr^d}}.
\]
Thus any sequence satisfying $r_K\to0$ and $Kr_K^d\to\infty$ yields vanishing worst-case regret under the estimators and assumptions of Theorem~\ref{thm:regret_linear_smoother}. The oracle resolution balances the two terms. However, the oracle resolution depends on the unknown smoothness parameters $(\beta, L)$. Nonetheless, in Proposition~\ref{pro:cross_validation} (Appendix~\ref{app:additional_theoretical_results}), we show that cross-validation can be used to select the resolution within a suitable growing sieve of linear smoothers while preserving the rate in
Theorem~\ref{thm:regret_linear_smoother}. Thus, the selected policy is minimax adaptive, mirroring similar results in nonparametric regression \citep{tsybakov2009introduction}.

\subsection{When Linear Smoothers are Minimax Optimal}

The next result shows that linear smoothers can be minimax optimal under Assumptions~\ref{ass:holder_smoothness}--\ref{ass:quasi_uniform_design}.
\begin{theorem}[Minimax Lower Bound]
\label{thm:minimax_lower_bound}
Suppose Assumptions~\ref{ass:gaussian_model}--\ref{ass:quasi_uniform_design} hold. Then there exists a constant $c>0$ such that for $K$ sufficiently large and $r \in [2 a_0 K^{-1/d}, r_0/2]$,
\begin{align*}
        \inf_{\widetilde\pi:\R^K \times [0, \infty)^K\to[0,1]^K}
        &\sup_{\substack{
        P_{g,\bm{\sigma}}\in\cP_K\\
        g\in\cH_d(\beta,L,H)
    }}
    R_K(P_{g,\bm{\sigma}},\widetilde\pi)
    \\
    &\qquad \quad \geq c
    \min\lt\{ H, Lr^\beta, \frac{\overline\sigma}{\sqrt{Kr^d}}\rt\}.
\end{align*}
In particular, for $r_K^\ast = \lt( \frac{d \overline{\sigma}}{2\beta L\sqrt K}\rt)^{\frac{2}{2\beta+d}}$,
\begin{align*}
         \inf_{\widetilde\pi:\R^K \times [0, \infty)^K\to[0,1]^K}
        &\sup_{\substack{
        P_{g,\bm{\sigma}}\in\cP_K\\
        g\in\cH_d(\beta,L,H)
    }}
    R_K(P_{g,\bm{\sigma}},\widetilde\pi) 
    \\
    &\, \geq c L^{d/(2\beta+d)}
    \overline\sigma^{2\beta/(2\beta+d)}
    K^{-\beta/(2\beta+d)}.
\end{align*}
\end{theorem}
The proof is given in Appendix~\ref{app:minimax_lower_bound}. The lower bound matches the minimax rates for binary classification \citep{yang1999minimax,audibert2007fast} under our observation model. Combined with Theorem~\ref{thm:regret_linear_smoother}, this establishes minimax-rate optimality of the partition-, kernel-, and B-spline-smoother policies over H\"older-smooth signals under quasi-uniform designs.

\section{Simulation Study}\label{sec:simulation}

\begin{figure*}[h!]
    \centering
    \includegraphics[width=\linewidth]{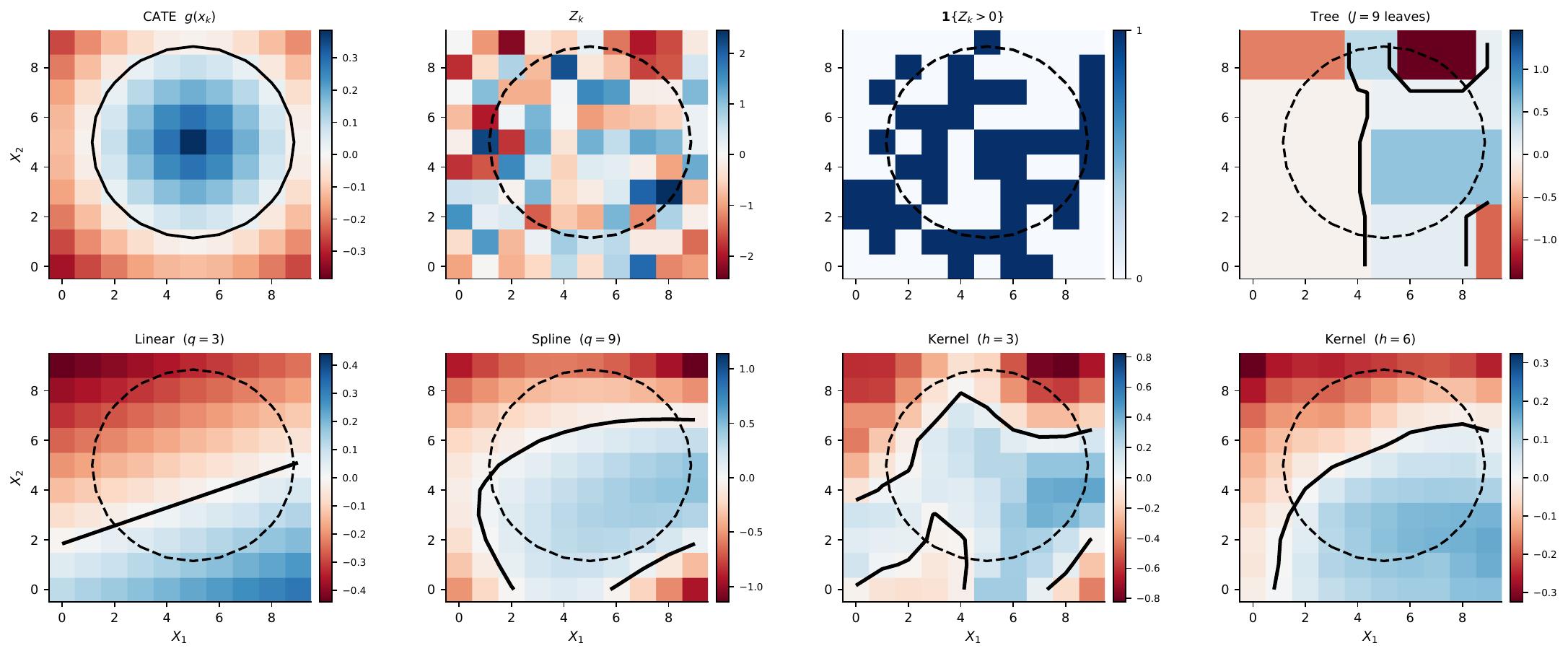}
    \caption{The first three panels show the true curved CATE $g$, one draw of the rescaled estimates $Z_k=\sqrt n\,\widehat\tau_k$, and the plug-in policy $\indicator{Z_k>0}$ that these estimates imply. The remaining panels show representative linear smoothers. The dashed line is the true frontier at which $g=0$, while the solid line represents the learner's decision boundary, $\widehat g_w=0$.}
    \label{fig:signal_curved}
\end{figure*}

We consider $K=100$ equally weighted tasks on a two-dimensional
grid, with unit-variance Gaussian noise and a curved, Lipschitz
signal $g_{\mathrm{crv}}$. The signal is centered, so treat-all and treat-none both have zero welfare. Appendix~\ref{app:simulation_design} provides additional detail on the simulation.

Figure~\ref{fig:signal_curved} illustrates this regime. The observed estimates display little structure, even though the underlying CATE is highly structured. Consequently, the plug-in policy $\indicator{Z_k > 0}$ mis-signs $45\%$ of tasks, against $50\%$ for a coin flip, and attains only $14\%$ of $V_{g_{\mathrm{crv}}}^\ast$. Comparing this pattern with Figure~\ref{fig:exp-1} shows that this simulation regime is relatively faithful to patterns observed in real-world data.

We consider three families of smoothers: axis-aligned trees with $J$ leaves, tensor-product B-splines of degree two with dimension $q$, and kernel smoothing with an Epanechnikov kernel of bandwidth $h$.

The remaining panels of Figure~\ref{fig:signal_curved} show representative smoothers at their oracle resolution predicted by Theorem~\ref{thm:regret_linear_smoother} ($J = q = 9$ and $h = 3$). Each adapt to the circular frontier and the best of them, the spline with $q = 9$, attains $46\%$ of $V_{g_{\mathrm{crv}}}^\ast$, more than three times the value of the plug-in policy. Table~\ref{tab:learners} reports the full comparison. In Figures~\ref{fig:tradeoff} and~\ref{fig:frontier} in Appendix~\ref{app:additional_simulation_results}, we illustrate the approximation--estimation tradeoff of Proposition~\ref{pro:smoothing_tradeoff} of the smoothers in terms of their resolution. Another simulation example is discussed in Appendix~\ref{app:linear_signal}.

\section{Empirical Application at Netflix}\label{sec:netflix}

\begin{figure*}[h!]
    \centering
    \includegraphics[width=\linewidth]{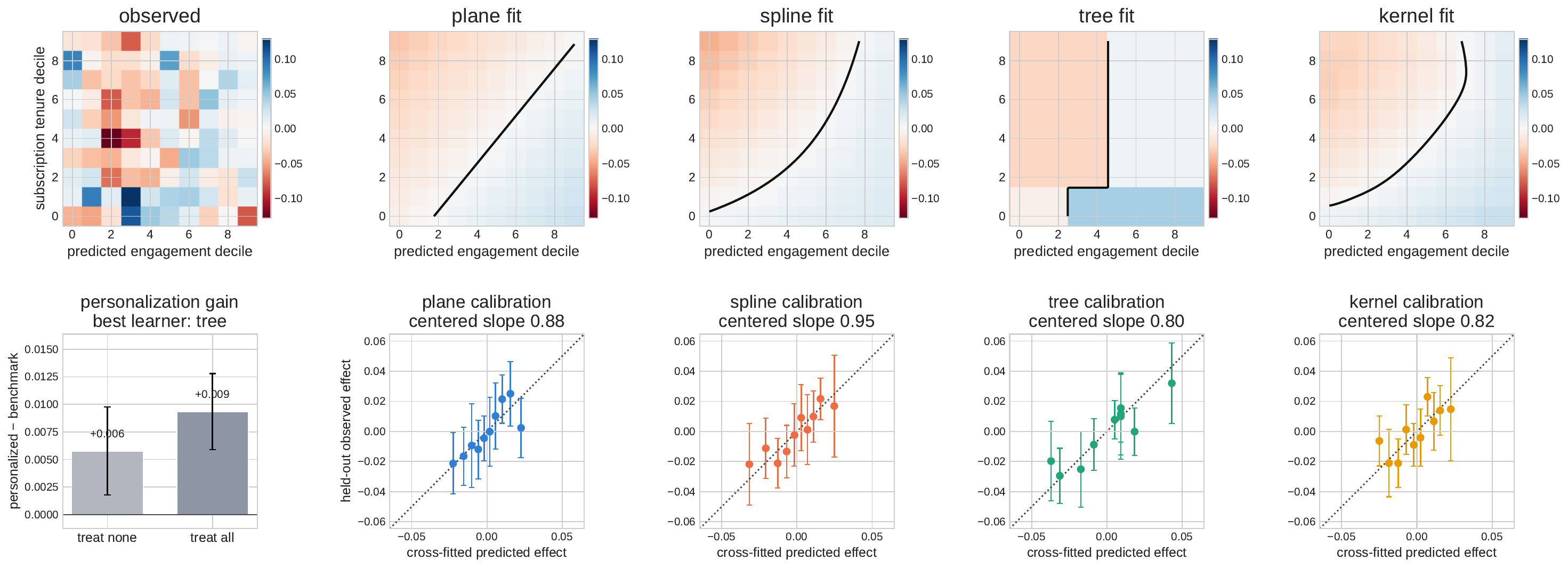}
    \caption{The top row shows the observed CATEs and full-sample fits of a plane, tensor-product B-spline, regression tree, and local-linear Epanechnikov kernel applied to Experiment 1. Each displayed fit uses the candidate selected most often across outer folds in nested CV.  The lower-left panel shows the average held-out policy gains over unpersonalized policies of the family with the lowest estimated policy loss.  The remaining panels show the calibration of cross-fitted CATE predictions after centering within each fold.}
    \label{fig:exp-1}
\end{figure*}

\begin{figure*}[h!]
    \centering
    \includegraphics[width=\linewidth]{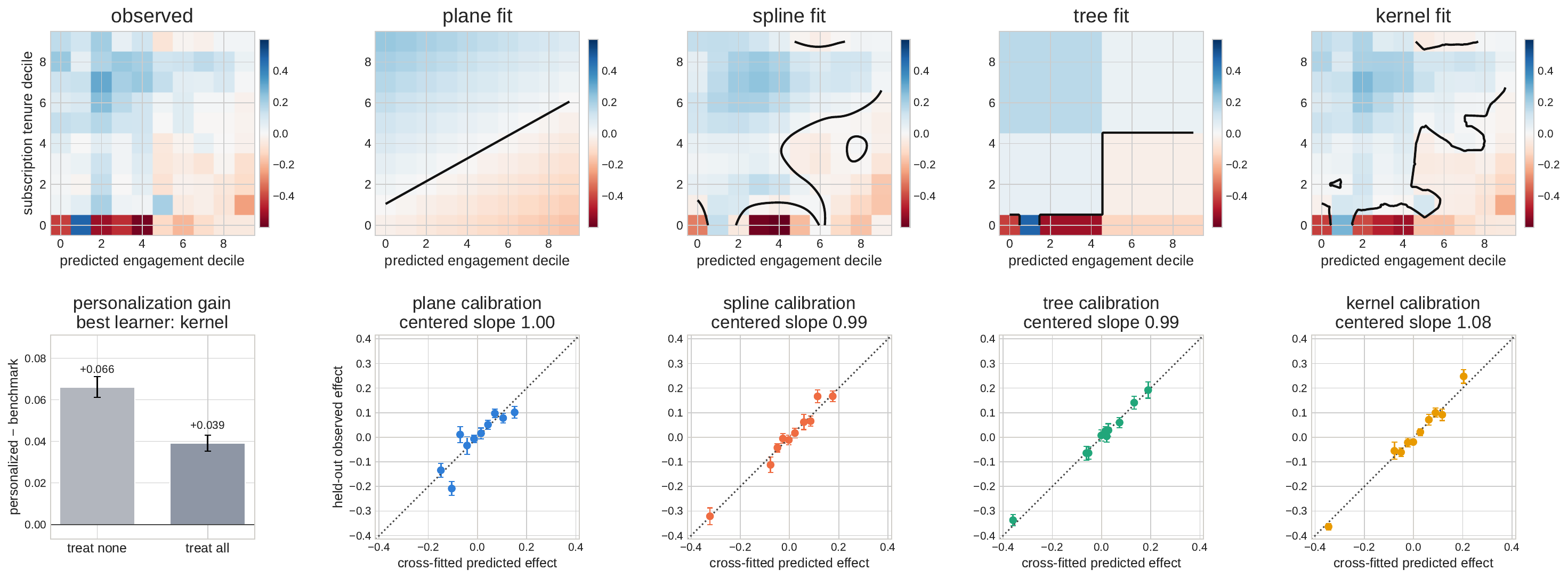}
    \caption{Panels are laid out as in Figure~\ref{fig:exp-1}, but for Experiment 2.}
    \label{fig:exp-2}
\end{figure*}

We demonstrate the empirical relevance of our framework by applying it to two large-scale real-world Netflix experiments.  For both experiments, we implement the proposed methodology as follows.  We discretize two pre-treatment covariates, predicted engagement level and length of subscription tenure, into ten deciles, yielding 100 distinct covariate cells.  Experimental units, which are divided into treatment and control arms, are further randomized to 10 folds.  We do not report the outcome metric for confidentiality reasons, but it can be interpreted as a measure of engagement frequency.  For each arm-fold-covariate value, we aggregate the data to the count, sum, and sum of squares of the outcome.  This allows us to collapse each experimental dataset, which consists of millions of observations, into a smaller dataset of 2,000 rows.

Let $k \in \{1,\ldots,100\}$ index covariate values, $f \in \{1,\ldots,10\}$ index folds, and $a\in\{0,1\}$ index arms. The CATE is $\tau_k=\mathbb{E}[Y(1)-Y(0)\mid X_k=k]$. In fold $f$, we estimate $\tau_k$ by $\hat\tau_{f,k}=\bar Y_{1,f,k}-\bar Y_{0,f,k}$ and calculate the empirical cell mass $p_{f,k}=(n_{1,f,k}+n_{0,f,k})/\sum_j(n_{1,f,j}+n_{0,f,j})$, which we use as learner weights. We fit four learner families to these 100 cell effects: a linear plane in the two decile indices, tensor-product B-splines, regression trees, and isotropic local-linear Epanechnikov kernel regression.

To jointly select the CATE learner and measure its corresponding policy value, we use nested cross-validation.  For outer fold $f$, candidate hyperparameter $\lambda$ in family $m$, and inner validation fold $g\ne f$, we fit on the other eight folds and predict $\hat\tau_{m,\lambda}^{(-f,-g)}(k)$. The inner-fold loss is the negative policy welfare $L_{f,g}(m,\lambda)=-\sum_{k=1}^{100}p_{g,k}\,\mathbf{1}\{\hat\tau_{m,\lambda}^{(-f,-g)}(k)>0\}\,\hat\tau_{g,k}$, motivated by Proposition~\ref{pro:cross_validation}. We select the simplest model that is within one standard error of the minimum average fold loss. We refit that candidate on all nine non-held-out folds and set the outer-fold policy to $\hat\pi_{m,f}(k)=\indicator{\hat\tau_m^{(-f)}(k)>0}$.  Lastly, for family $m$ and comparator $c=0$ (treat none) or $c=1$ (treat all), we estimate the welfare gain by $\hat\Delta_m(c) = \frac{1}{10}\sum_{f=1}^{10}\sum_{k=1}^{100}p_{f,k}(\hat\pi_{m,f}(k)-c)\hat\tau_{f,k}$ with variance $\widehat{\mathrm{SE}}_m^2(c)=\frac{1}{10^2}\sum_{f,k}p_{f,k}^2(\hat\pi_{m,f}(k)-c)^2\bigl(s_{1,f,k}^2/n_{1,f,k}+s_{0,f,k}^2/n_{0,f,k}\bigr)$.

\subsection{Experiment 1}
This experiment allocated approximately 9.41 million users to treatment and 10.54 million users to control.  This experiment modified the weights used in recommender system training to emphasize newer titles.  Results are shown in Figure~\ref{fig:exp-1}.

As the first panel of Figure~\ref{fig:exp-1} shows, the grid of observed CATEs over the 100 covariate values is very noisy.  However, all four linear smoothing families yield well-calibrated and interpretable CATE functions that indicate that treatment is more effective for highly-engaged and less-tenured units.  The selected tree-family policy significantly lifts engagement over both unpersonalized policies (treat none and treat all).

\subsection{Experiment 2}

This experiment allocated approximately 4.73 million users to treatment and 10.65 million users to control.  The treatment augmented the features available for newly launching titles.  Results are shown in Figure~\ref{fig:exp-2}.

The raw CATEs, which are shown in the first panel of Figure~\ref{fig:exp-2}, show somewhat more structure, but still with substantial noise.  All model families yield a well-calibrated and interpretable CATE estimator, which show that the treatment is more effective for more-tenured but less-engaged users.  The selected kernel-family policy significantly lifts engagement over both unpersonalized policies.

\section{Conclusion}

We study the problem of policy learning in the weak-signal setting, in which treatment effects are small and on the scale of measurement error. We establish the general impossibility of learning the optimal policy and the slowness of learning its value in this regime.  However, when treatment effects are smooth, linear smoothing estimators that exploit this structure can be minimax optimal. We show how to derive the optimal resolution of such learners and apply our framework to two large-scale real-world experiments at Netflix. Our results demonstrate the promise of our framework for policy learning in digital experiments.

\clearpage

\subsection*{AI Use Statement}

We used generative AI tools to assist with data analysis, code and figure generation, manuscript editing (including grammar checks, spelling errors, and improvements to clarity), the derivation of selected critical proof steps, and verifying proofs. Generative AI was not used to generate research ideas or for substantive text generation. We checked the code and results and reviewed the proofs line by line. We take responsibility for the final content of this work, including text, claims, and artifacts produced with the aid of generative AI.

\newpage

\bibliography{references}
\bibliographystyle{plainnat}

\clearpage
\appendix
\thispagestyle{empty}

\onecolumn
\begin{center}
{\Large\bfseries Supplementary Materials}
\end{center}

\section{Gaussian Limit Experiments}\label{app:limit_experiments}
In Section~\ref{sec:setup}, we formulated decision rules as functions
of task-level data, $\widetilde\pi(\widehat\tau,\widehat\sigma)$,
and assumed a Gaussian sequence model (Assumption~\ref{ass:gaussian_model}). This can be justified by local asymptotics. We provide a discussion here and refer the reader to \citet{van2000asymptotic,hirano2009asymptotics} for more details.

For each task $x_k$, denote the individual-level observations by
$D_{k,n}=\{(A_{ki},Y_{ki}):i=1,\ldots,n\}$, where $A_{ki}\in\{0,1\}$ is the randomized treatment assignment and
$Y_{ki}$ is the observed outcome. For simplicity, we consider independent samples of $n$ units
for each task. We assume that $n$ is large. Writing $D_n=(D_{1,n},\ldots,D_{K,n})$ for the full data,
a general decision rule takes the form
$\widetilde\pi_n(D_n)\in[0,1]^K$.

Let $\widehat \tau_k$ be an estimator of $\tau_k$ based on $D_{k,n}$, for example the difference in means estimator $\widehat \tau_k = \overline Y_{1, k} - \overline Y_{0, k}$. By the central limit theorem, under regularity conditions, 
\[
    \sqrt{n} (\widehat \tau_k - \tau_k) \rightsquigarrow \cN(0, \sigma_k^2),
\]
in distribution, justifying Assumption~\ref{ass:gaussian_model}.

To justify task-level decision rules, by Assumption~\ref{ass:weak_signal}, $\tau_k = \frac{h_k}{\sqrt{n}}$, so that
\[
    Z_{k, n} := \sqrt{n} \widehat \tau_k \rightsquigarrow \cN(h_k, \sigma_k^2).
\]
Writing $Z_n = (Z_{1, n}, \ldots, Z_{k, n})$, then $Z_n \rightsquigarrow \cN(\bm h, \operatorname{diag}(\sigma_1, \ldots, \sigma_K))$, for $\bm h = (h_1, \ldots, h_K)$. 

Let $Q_{n, \bm h}$ be the distribution of $D_n$ in a regular submodel and $\widehat \tau_k$ a regular, asymptotically efficient estimator in this submodel. Let $\phi_{\bm h}$ denote
the corresponding Gaussian limiting density of $Z_n$. The likelihood-ratio expansion \citep[Theorem~7.2 and Lemma~8.14]{van2000asymptotic} gives
\[
    \log\lt(\frac{d Q_{n, \bm h}}{d Q_{n, \bm 0}}\rt) (D_n) = \log\lt(\frac{\phi_{\bm h}(Z_n)}{\phi_{\bm 0}(Z_n)}\rt) + o_{Q_{n,\bm 0}}(1), 
\]
i.e., to distinguish $\bm h$ from $\bm 0$ all the evidence from individual units is summarized (up to first order) by $Z_n$. In particular, this is all the information needed to determine a cell treatment decision (sign of $\bm h$)

To connect this to decision rules, suppose that $\widetilde \pi_n(D_n)$ converges in distribution under every local alternative. By \citet[Theorem~9.3]{van2000asymptotic} its limiting distribution can be reproduced by a randomized policy
\[
    \widetilde \pi_{\infty}(Z, U), \qquad Z \sim \cN(\bm h, \operatorname{diag}(\sigma_1, \ldots, \sigma_K)), U \sim \operatorname{Unif}[0, 1].
\]
Thus, in the limit, every unit-level policy can be replicated by a Gaussian experiment and uniform noise. Finally, since our welfare criterion is based on expectations, for each randomized limiting policy, we can find a corresponding average policy $\overline{\pi}_k(z) = \E_U\lt[\widetilde \pi_{\infty, k}(z, U) \rt]$ with the same expected welfare. With this justification, we can restrict to limiting decision rules $\overline\pi(Z)\in[0,1]^K$. Since noise levels are unknown, we let decision rules additionally depend on estimates of them too.

\section{Linear Smoothers} \label{app:smoothers}

In Section~\ref{sec:linear_smoothing} we focus on three classes of linear smoothers, each defined by its weight functions and a \emph{resolution} parameter $r$, which controls the scale of smoothing. A small $r$ corresponds to less smoothing. Smaller $r$ allows for better approximation but typically increases variance. Larger $r$ pools more information across neighboring tasks, reducing variance at the cost of (potentially) greater bias.

\begin{definition}[Partition Pooling]\label{def:partition_smoothing}
Let $r> 0$ and let $\cA_r$ be a finite partition of $\cX$ such that every cell contains a task and $\max_{A \in \cA_r} \sup_{x, y \in A}\|x - y\|_2 \leq r$. Define
\[
    w_{r, j}^{\mathrm P}(x)
    :=
    \frac{p_j\indicator{x_j\in A_r(x)}}
    {\sum_{\ell:x_\ell\in A_r(x)}p_\ell},
\]
where $A_r(x)$ denotes the cell containing $x$.
\end{definition}
Definition~\ref{def:partition_smoothing} includes tree-based pooling, where the partition cells are the leaves of a tree and $w_r^{\mathrm P}$ averages within each leaf. Under the assumptions of Theorem~\ref{thm:regret_linear_smoother}, a tree at resolution $r$ has $J_r\asymp r^{-d}$ leaves.

\begin{definition}[Kernel Smoothing]\label{def:kernel_smoothing}
Let $\sK:\R^d\to[0,\infty)$ be bounded, supported on $B(0,1)$,
continuous at the origin, and satisfy $\sK(0)>0$. For $r>0$, define
\[
    w_{r, j}^{\mathrm K}(x)
    :=
    \frac{\sK\lt((x_j-x)/r\rt)}
    {\sum_{\ell=1}^K\sK\lt((x_\ell-x)/r\rt)}.
\]
\end{definition}
Partition pooling imposes a common action within cells, while kernel smoothing can lead to a different decision at every task location. We associate the bandwidth of the kernel with its resolution.

\begin{definition}[Series Regression]
\label{def:series_regression}
Fix $r>0$. Let $\bm a=(a_1,\ldots,a_K)$ be fitting weights satisfying
$a_k>0$ and $\sum_{k=1}^K a_k=1$, and let
$b_r:\cX\to\R^{q_r}$ be a vector of basis functions with $q_r\leq K$ and $q_r \asymp r^{-d}$.  Assume that the constant
function belongs to the span of $b_r$ and that 
$G_{r,\bm a}  := \sum_{k=1}^K a_k b_r(x_k)b_r(x_k)^\top$ is nonsingular. Define
\[
    w_{r, \bm a, j}^{\mathrm S}(x) := a_j b_r(x)^\top G_{r,\bm a}^{-1} b_r(x_j).
\]
\end{definition}
The weights in Definition~\ref{def:series_regression} are induced by the weighted least-squares estimator
\[
    \widehat\theta_{r,\bm a}
    :=
    \argmin_{\theta\in\R^{q_r}}
    \sum_{k=1}^K
    a_k(Z_k-b_r(x_k)^\top\theta)^2,
\]
with resulting signal estimator $\widehat g_{w_{r,\bm a}^{\mathrm S}}(x) := b_r(x)^\top\widehat\theta_{r,\bm a} = \sum_{j=1}^K  w_{r,\bm a,j}^{\mathrm S}(x)Z_j$.
By requiring that the constant function is in the span of the basis function, we guarantee $\sum_{j=1}^K w_{r,\bm a,j}^{\mathrm S}(x)=1$

Choosing $a_k=p_k$ gives the welfare-weighted series smoother and makes a cell-indicator basis coincide exactly with partition pooling. Choosing $a_k=1/K$ gives ordinary least squares, while choosing $a_k\propto \widehat\sigma_k^{-2}$ gives inverse-variance-weighted least squares.

\subsection{Smoothers for Theorem~\ref{thm:regret_linear_smoother}}\label{app:smoother_conditions}

Theorem~\ref{thm:regret_linear_smoother} applies smoothers induced by the following weights:

\begin{enumerate} [label=(\alph*)]
    \item \emph{Partition pooling}: Let $\cA_r$ be as in Definition~\ref{def:partition_smoothing}, and suppose there exists a constant $c_{\mathrm P}>0$ such that $\min_{A\in\cA_r} Q_K(A) \geq c_{\mathrm P} r^d$. Let $w_r^{\mathrm P}$ denote the weights induced by this partition.
   \item \emph{Kernel smoothing}: Let $\sK$ be as in
    Definition~\ref{def:kernel_smoothing} inducing $w_r^{\mathrm K}$.
    \item \emph{Spline regression}: Suppose $\cX\subseteq[0,1]^d$ and let
    $b_r$ be a normalized tensor-product B-spline basis with knot spacing of order $r$. That is, $b_{r,\nu}(x)\ge 0$ and $\sum_{\nu=1}^{q_r}b_{r,\nu}(x)\equiv 1$ for $x \in \cX$. Let $a_k = p_k$ and impose the conditions in Definition~\ref{def:series_regression}. Let $w_{r}^{\mathrm S}$ denote the induced weights.
\end{enumerate}
The assumption on partition pooling is mild. In particular, it is satisfied under Assumption~\ref{ass:quasi_uniform_design} whenever, around each task covariate, one can inscribe a ball of radius $\rho r$, $\rho \in (0, 1/2)$, that is contained in the same partition cell.

\section{Additional Simulation Results} \label{app:additional_simulation_results}
This appendix collects and extends the results referenced in Section~\ref{sec:simulation}.

\subsection{Simulation Design of Section~\ref{sec:simulation}}\label{app:simulation_design}

We place $K = 100$ tasks on a regular grid at each covariate combination
$x_k = (X_1, X_2)$ with $X_1, X_2 \in \{0.05, 0.15, \ldots, 0.95\} \subseteq [0,1]^2$ and $d = 2$. Task are weighted equally, $p_k = 1/K$, and
noise is homoskedastic $\sigma_k = 1$. Adjacent tasks are spaced $K^{-1/d}=0.1$ apart, which we use as the unit of resolution. Thus Assumption~\ref{ass:quasi_uniform_design} holds by construction. 

For the curved signal, define 
\[
    g_0(x)=0.3-\|x-c\|_2,
\]
where $c=(0.55,0.55)$ and center it as
\[
    g_{\mathrm{crv}}(x) = g_0(x)-\frac{1}{K}\sum_{j=1}^K g_0(x_j).
\]
The centered signal is H\"older-smooth with $\beta=1$ and $L=1$. Since $\sum_{k=1}^K p_k g_{\mathrm{crv}}(x_k)=0$, the treat-all and treat-none policies both have value zero. The resulting signal satisfies $H=0.39$, so the design lies within Assumption~\ref{ass:weak_signal}.

Recall that we considered three types of smoothers: axis-aligned trees with $J$ leaves, tensor-product B-splines of degree two with dimension $q$ (where $q=3$ corresponds to ordinary least squares with features $(1,X_1,X_2)$), and kernel smoothing with an Epanechnikov kernel of bandwidth $h$. To account for grid spacing, we measure resolution in grid units, $\tilde r = K^{1/d}r$ for $r\in[0,1]$. For trees and splines, we set $\tilde r=K^{1/d} J^{-1/d}$ and $ \tilde r=K^{1/d} q^{-1/d}$, respectively, with $J\in\{1,4,9,16,25,36,49\}$ and $q\in\{1,3,9,16,25,36,49\}.$ For kernel smoothing, we set $\tilde r=h$ and vary $h\in[1,10]$.

The three families agree closely at the grid-unit resolution predicted by Theorem~\ref{thm:regret_linear_smoother}, $\tilde r_K^\ast \approx 3.16$. The nearest candidate choices are $J=9$ leaves for trees, $q=9$ for splines, and $h=3$ for kernels.

\subsection{Linear Signal}\label{app:linear_signal}

Alongside the curved signal in Section~\ref{sec:simulation}, we consider a linear signal on the same design,
\[
    g_{\mathrm{lin}}(x) = (x - \bar{x})^\top u
\]
where $u = (1,1)/\sqrt2$ and $\bar x = \frac{1}{K} \sum_{k = 1}^K x_k$. This is again a H\"older smooth function with $\beta = 1$ and $L = 1$. It satisfies $H = 0.64$ and thus sits within Assumption~\ref{ass:weak_signal}.

Figure~\ref{fig:signal_linear} illustrates the regime. Even through the overall signal is almost twice as strong as in Section~\ref{sec:linear_smoothing} the plug-in policy still mis-labels $41\%$ of tasks and attains only $27\%$ of $V_{g_{\mathrm{lin}}}^\ast$.

The linear smoothers perform quite well in this regime, with the linear regression at $q=3$ performing best and, unsurprisingly, attaining $88\%$ of $V_{g_{\mathrm{lin}}}^\ast$. The performance of the remaining learners is reported in Table~\ref{tab:learners}.

\begin{figure}[t]
    \centering
    \includegraphics[width=\linewidth]{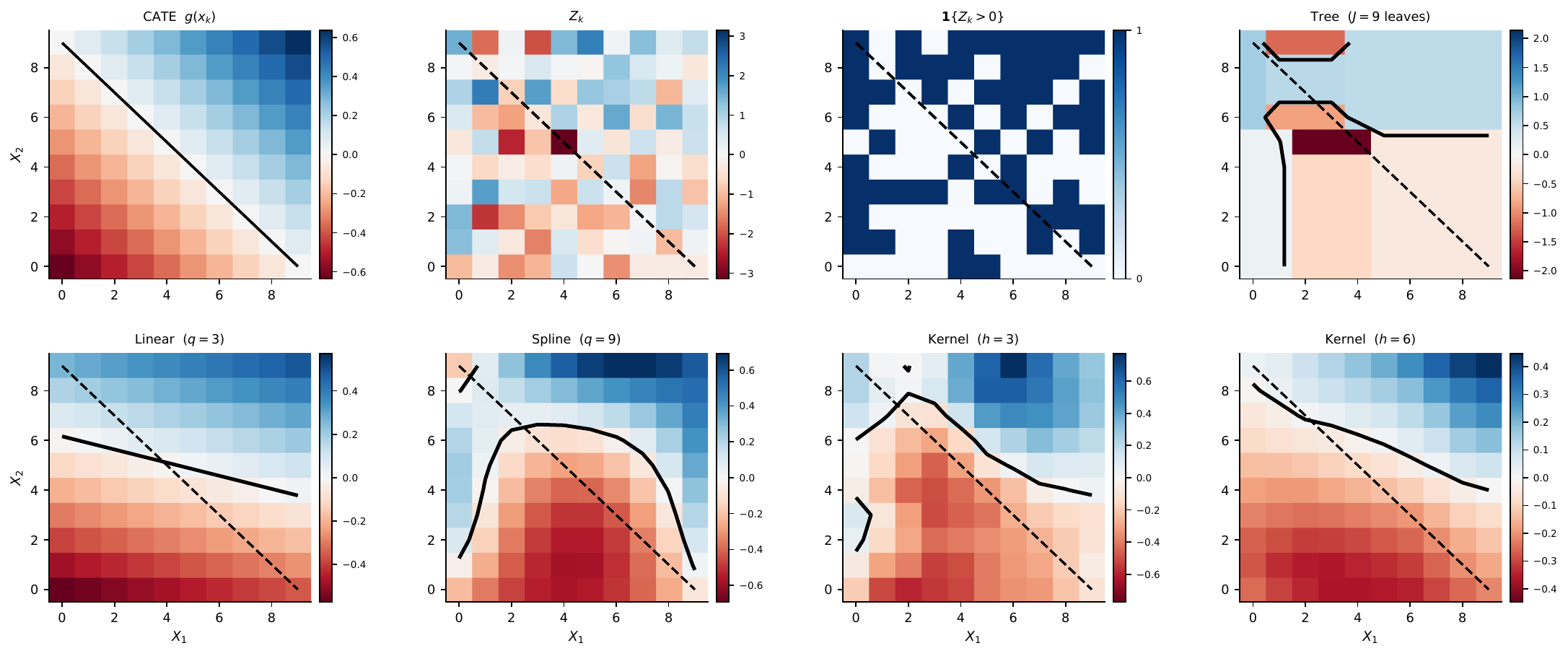}
    \caption{The first three panels show the true linear CATE $g$, one draw of the rescaled estimates $Z_k=\sqrt n\,\widehat\tau_k$, and the plug-in policy $\indicator{Z_k>0}$ that these estimates imply. The remaining panels show representative linear smoothers. The dashed line is the true frontier at which $g=0$, while the solid line represents the learner's decision boundary, $\widehat g_w=0$.}
    \label{fig:signal_linear}
\end{figure}

\subsection{Comparison of Learners}\label{app:learner_comparison}

Table~\ref{tab:learners} compares representative learners from each family at three matched resolutions under both signals. Because both signals are centered, $W/V^\ast=0$ means that a learner does no better than a constant policy, while $W/V^\ast=1$ means that it matches the oracle. In the table and figures, partitions correspond to fixed axis aligned partitions with $J$ cells.

The table highlights the approximation--estimation tradeoff of Proposition~\ref{pro:smoothing_tradeoff}. Within each family, if the resolution is too coarse, approximation error dominates. However, if the resolution is too fine, estimation error dominates. We observe that kernel smoothing adapts well to both signal regimes.

Figure~\ref{fig:tradeoff} further illustrates this point and also plots the bound from Proposition~\ref{pro:smoothing_tradeoff}. The bound is close for series regression and less close for trees. Overall, the figure shows that balancing approximation and estimation error predicts the optimal resolution well.

Figure~\ref{fig:frontier} additionally illustrates the performance of our linear smoothers in terms of their parameter complexity. Especially for the curved signal, striking the right balance in parameter complexity, and hence in resolution, is crucial for good performance.

\begin{table}[t]
\centering
\small
\setlength{\tabcolsep}{5pt}

\begin{tabular}{@{}lc@{\hskip 12pt}rr@{\hskip 12pt}rr@{}}
\toprule
& & \multicolumn{2}{c}{Linear signal} & \multicolumn{2}{c}{Curved signal}\\
\cmidrule(lr){3-4}\cmidrule(lr){5-6}
Learner & $\tilde r$ & $R_K$ & $W/V^\ast$ & $R_K$ & $W/V^\ast$\\
\midrule
Plug-in              & $1$    & $0.085$ & $0.27$ & $0.052$ & $0.14$\\
Pooled               & $10$   & $0.117$ & $0.00$ & $0.060$ & $0.00$\\
\midrule
Tree ($J = 4$)       & $5.0$  & $0.035$ & $0.70$ & $0.034$ & $0.43$\\
Tree ($J = 9$)       & $3.3$  & $0.040$ & $0.65$ & $0.034$ & $0.44$\\
Tree ($J = 25$)      & $2.0$  & $0.063$ & $0.46$ & $0.043$ & $0.30$\\
\midrule
Partition ($J = 4$)  & $5.0$  & $0.035$ & $0.70$ & $0.056$ & $0.07$\\
Partition ($J = 9$)  & $3.3$  & $0.041$ & $0.65$ & $0.045$ & $0.25$\\
Partition ($J = 25$) & $2.0$  & $0.060$ & $0.49$ & $0.046$ & $0.24$\\
\midrule
Kernel ($h = 5$)     & $5.0$  & $0.024$ & $0.80$ & $0.041$ & $0.32$\\
Kernel ($h = 3$)     & $3.3$  & $0.033$ & $0.72$ & $0.038$ & $0.37$\\
Kernel ($h = 2$)     & $2.0$  & $0.048$ & $0.59$ & $0.042$ & $0.31$\\
\midrule
Series ($q = 3$)     & $5.8$  & $\mathbf{0.014}$ & $\mathbf{0.88}$ & $0.056$ & $0.08$\\
Series ($q = 9$)     & $3.3$  & $0.033$ & $0.72$ & $\mathbf{0.033}$ & $\mathbf{0.46}$\\
Series ($q = 25$)    & $2.0$  & $0.059$ & $0.50$ & $0.042$ & $0.30$\\
\bottomrule
\end{tabular}
\caption{Exact welfare regret and proportional welfare for representative learners. The resolution $\tilde r$ is measured in units of $K^{-1/d}$, so $\tilde r=1$ corresponds to the plug-in rule and $\tilde r=10$ corresponds to full pooling, that is, a single treatment decision is made for all cells jointly.}
\label{tab:learners}
\end{table}

\begin{figure}[t]
    \centering
    \includegraphics[width=\linewidth]{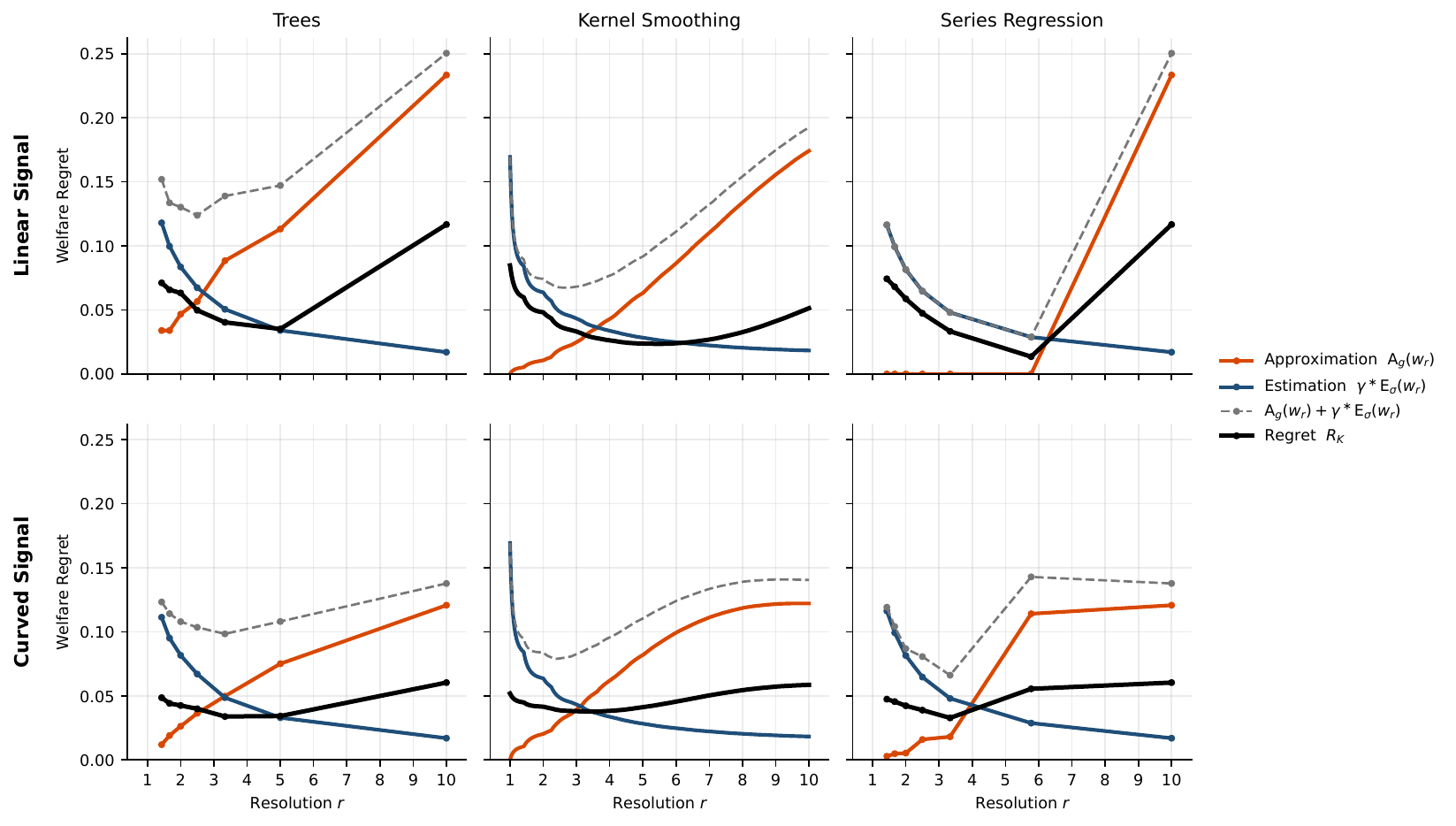}
    \caption{Approximation error $\sA_g(w_r)$, estimation error $\gamma^\ast\sE_{\bm\sigma}(w_r)$, their sum, and exact regret $R_K$ as functions of resolution $\tilde r$. Rows correspond to signal regimes, while columns correspond to the smoother families considered in the main text.}
    \label{fig:tradeoff}
\end{figure}

\begin{figure*}[t]
    \centering
    \includegraphics[width=\linewidth]{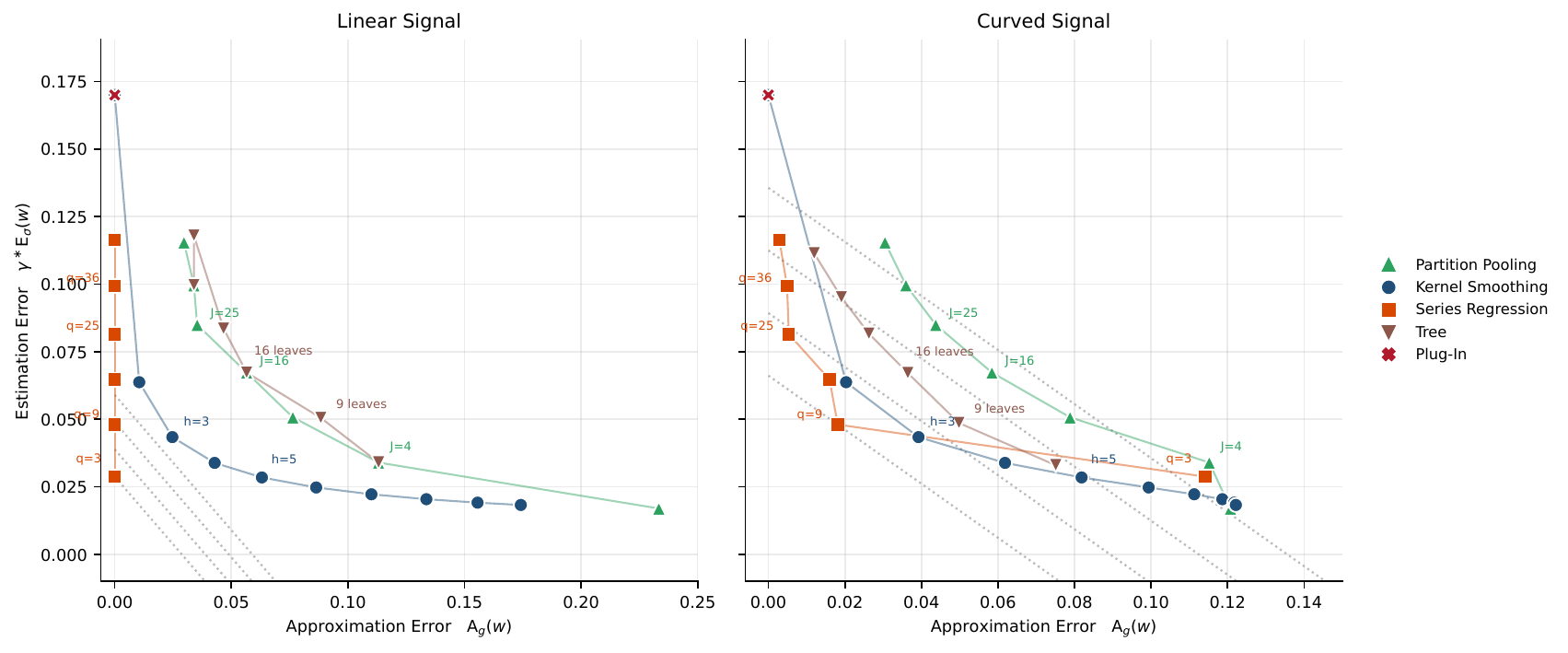}
    \caption{Every candidate is one point in terms of approximation error $\sA_g(w)$ and estimation error $\gamma^\ast\sE_{\bm\sigma}(w)$. Note that the axis are on different scales. Better smoothers lie closer to the bottom-left corner.}
    \label{fig:frontier}
\end{figure*}

\clearpage
\newpage 
\section{Additional Theoretical Results}\label{app:additional_theoretical_results}

\begin{proposition}[Welfare Decomposition for Linear Smoothers]\label{pro:welfare_decomposition}
    Suppose Assumptions~\ref{ass:gaussian_model}--\ref{ass:weak_signal} hold. Let $\widehat g_w$ be a linear smoother. Then,
    \begin{align*}
         \sqrt{n} W_{P_{g, \bm{\sigma}}}(\widetilde\pi_w) = \bar g_p \bar \psi_{w, p} + \Cov_p(g, \psi_w), 
    \end{align*}
    where $\bar g_p = \sum_{k = 1}^K p_k g(x_k)$, $\psi_{w, k} = \E_{g, \bm{\sigma}}\lt[\widetilde\pi_{w, k}(Z)\rt] = \Phi\lt(\frac{m_w(x_k)}{s_w(x_k)}\rt)$, $\bar \psi_{w, p} = \sum_{k = 1}^K p_k \psi_{w, k}$ and
    \[
        \Cov_p(g, \psi_w) = \sum_{k = 1}^K p_k (g(x_k) - \bar g_p) (\psi_{w, k} - \bar \psi_{w, p}) = \frac{1}{2} \sum_{k = 1}^K \sum_{l = 1}^K p_k p_l (g(x_k) - g(x_l)) (\psi_{w, k} - \psi_{w, l}).
    \]
\end{proposition}
The proof is given in Appendix~\ref{app:welfare_decomposition}.
Proposition~\ref{pro:welfare_decomposition} decomposes welfare into a baseline term and a covariance term. The baseline term captures what would be achieved by treating every task with probability $\bar\psi_{w,p}$. The covariance term measures the personalization capacity of the policy. It is larger when tasks with larger treatment effects are more likely to be treated.

\begin{proposition}[Cross-Validation]\label{pro:cross_validation}
Suppose Assumptions~\ref{ass:gaussian_model}--\ref{ass:weak_signal} hold. Assume further that $\max_k p_k = O(1/K)$. Suppose for each task $x_k$, the data is split into $F$ folds and satisfies 
\[
    Z_k^{(f)} \overset{\mathrm{ind}}{\sim} \cN(h_k, F \sigma_k^2), 
\]
for $k \in [K], f \in [F]$. Denote $Z_k^{(-f)} = \frac{1}{F - 1} \sum_{j \neq f} Z_k^{(j)}$ and $Z_k = \frac{1}{F} \sum_{f = 1}^F Z_k^{(f)}$. Let $\Lambda_K$ be a finite set of resolutions. For each $r \in \Lambda_K$, define the linear smoothing policies $\widetilde \pi_{w_r, k}(z) = \indicator{\sum_{j = 1}^K w_{r, j}(x_k) z_j > 0}$. Select
\[
    \widehat r \in \argmax_{r \in \Lambda_K} \frac{1}{F} \sum_{f = 1}^F \sum_{k = 1}^K p_k Z_k^{(f)} \widetilde\pi_{w_r,k}(Z^{(-f)}).
\]
Then, 
\[
    R_K(P_{g, \bm{\sigma}}, \widetilde\pi_{w_{\widehat r}}) \lesssim  \min_{r\in\Lambda_K} R_K(P_{g, \bm{\sigma}_F}, \widetilde\pi_{w_r}) + \sqrt{\frac{\log(|\Lambda_K|)}{K}},
\]
where $\bm\sigma_F=\sqrt{F/(F-1)}\bm\sigma$.
\end{proposition}
The proof is given in Appendix~\ref{app:cross_validation}. Proposition~\ref{pro:cross_validation} justifies our use of cross-validation by maximizing welfare in Section~\ref{sec:netflix}. Under Assumption~\ref{ass:gaussian_model}, $Z_k=\frac{1}{F}\sum_{f=1}^F Z_k^{(f)}\sim\cN(h_k,\sigma_k^2)$, matches our observation model.
The selected smoother achieves the regret of the best candidate resolution on $\Lambda_K$ with an additional selection cost. If $|\Lambda_K|$ grows at most polynomially in $K$, this cost is $o(K^{-\beta/(2\beta+d)})$, in which case cross-validation preserves the rate in Theorem~\ref{thm:regret_linear_smoother}.

\section{Connection to Binary Classification and Regression}\label{app:binary_classification}
Choosing the optimal treatment decision for each task is analogous to binary classification \citep{yang1999minimax, audibert2007fast} and, for plug-in rules, bounded by $L^1$ regression error \citep{stone1982optimal}. This connection is well established in the policy-learning literature \citep[e.g.,][]{kitagawa2018should, luedtke2020performance}. In particular, our scaled welfare regret is proportional to a weighted excess classification risk. However, we study a nonuniformly weighted, fixed-design, heteroskedastic Gaussian observation model, instead of observing i.i.d. binary labels directly.

Let $Q_K = \sum_{k = 1}^K p_k \delta_{x_k}$ denote the task measure. Consider the auxiliary binary classification problem
\begin{align*}
    X &\sim Q_K
    \\
    Y \mid X = x &\sim \text{Ber}(\eta_g(x)),
\end{align*}
where $\eta_g(x) = \frac{1}{2} \lt(1 + \frac{g(x)}{H}\rt) \in [0, 1]$ for $g \in \cH_d(\beta, L, H)$. The Bayes classifier is given by $\pi^\ast_g(x) = \indicator{\eta_g(x) > \frac{1}{2}} = \indicator{g(x) > 0}$.

We do not observe labels from this auxiliary model. Instead, we observe $(Z, \widehat \sigma) \sim P_{g, \bm \sigma}$, where $Z_k=g(x_k)+\sigma(x_k) \epsilon_k$ with $\epsilon_k\overset{\mathrm{ind}}{\sim}\cN(0,1)$, for all $k \in [K]$. Write $\widehat\pi := \widetilde\pi(Z, \widehat \sigma)$ for the fitted classifier, a map from $\cX$ to $\{0,1\}$. Its expected excess risk in
the auxiliary classification problem is \citep[e.g.,][]{audibert2007fast}
\begin{align*}
    \mathcal{E}(P_{g, \bm{\sigma}}, \widetilde\pi) &= \E_{P_{g, \bm{\sigma}}} \lt[\E_{X}[|2 \eta_g(X) - 1| \indicator{\widehat\pi(X) \neq \indicator{g(X) > 0}}]\rt]
            \\
            &=  \E_{P_{g, \bm{\sigma}}} \lt[\frac{1}{H}\sum_{k=1}^K p_k |g(x_k)|  \indicator{\widehat\pi(x_k) \neq \indicator{g(x_k) > 0}}\rt]
            \\
            &= \frac{1}{H}\sum_{k=1}^K p_k |g(x_k)|  P_{g, \bm{\sigma}} \lt(\widehat\pi(x_k) \neq \indicator{g(x_k) > 0}\rt)
            \\
            &= \frac{1}{H}R_K(P_{g, \bm{\sigma}}, \widetilde\pi),
\end{align*}
Thus the two problems have the same loss but the observation model differs. Moreover, in the auxiliary classification model, $g$ determines both the conditional mean and variance. In our
Gaussian model, the variances must be controlled independently.

Furthermore, for any estimator $\widehat g$ and its induced plug-in policy
$\widehat\pi(X)=\indicator{\widehat g(x)>0}$, a decision error implies that $|g(x_k)|\leq|\widehat g(x_k)-g(x_k)|$. Hence,
\[
R_K(P_{g, \bm{\sigma}},\widetilde\pi)
\leq \E_{P_{g, \bm{\sigma}}}\lt[ \sum_{k=1}^K p_k|\widehat g(x_k)-g(x_k)| \rt] = \E_{P_{g, \bm{\sigma}}}\|\widehat g-g\|_{L^1(Q_K)}.
\]
Thus global $L^1$ regression rates provide an upper bound on the welfare regret.

\section{Mathematical Proofs}

\subsection{Proof of Theorem~\ref{thm:minimax_regret}}\label{app:minimax_regret}

We first establish the minimax claim before proving the Bayes identity.
Collect $Z = (Z_1, \ldots, Z_K)$, $\widehat \sigma = (\widehat \sigma_1, \ldots, \widehat \sigma_K)$ and write $\bm{\overline{\sigma}} := (\overline \sigma, \ldots, \overline \sigma)$. Fix a probability law $Q_0$ on $[0, \infty)^K$. Let $\cP_K^0 \subseteq \cP_K$ consist of the laws under which $\widehat \sigma \sim Q_0$ and $\widehat \sigma \indep Z$. On $\mathcal P_K^0$, our variance estimates $\widehat \sigma$ are uninformative about $\bm h$. That is, every rule $\widetilde\pi(Z, \widehat \sigma)$ has the same regret under every law in $\cP_K^0$ as $\overline \pi(z) := \int \widetilde\pi(z, s) Q_0(ds)$. Thus, on $\cP^0_K$, we can restrict to rules that ignore $\widehat \sigma$.

Fix an arbitrary $t \in (0, H]$. Since $\{-t, t\}^K \subseteq [-H, H]^K$ and $\bm{\overline{\sigma}} \in [\underline \sigma, \overline{\sigma}]^K$, we bound 
    \begin{align*}
        \inf_{\widetilde\pi: \R^K \times [0, \infty)^K \to [0, 1]^K} \sup_{P_{\bm h, \bm \sigma} \in \mathcal P_K} R_K(P_{\bm h, \bm \sigma}, \widetilde\pi) &\geq  \inf_{\widetilde\pi: \R^K \times [0, \infty)^K \to [0, 1]^K} \sup_{P_{\bm h, \bm{\overline{\sigma}}} \in \mathcal P_K^0} R_K(P_{\bm h, \bm{\overline{\sigma}}}, \widetilde\pi)
        \\
        &=  \inf_{\widetilde\pi: \R^K \to [0, 1]^K} \sup_{P_{\bm h, \bm{\overline{\sigma}}} \in \mathcal P_K^0} R_K(P_{\bm h, \bm{\overline{\sigma}}}, \widetilde\pi)
        \\
        &=  \inf_{\widetilde\pi: \R^K \to [0, 1]^K} \sup_{P_{\bm h, \bm{\overline{\sigma}}} \in \mathcal P_K^0} \E_{P_{\bm h, \bm{\overline{\sigma}}}} \lt[\sum_{k = 1}^K p_k h_k (\indicator{h_k > 0} - \widetilde\pi_k(Z))\rt]
        \\
        &\geq  \inf_{\widetilde\pi: \R^K \to [0, 1]^K} \max_{\bm h  \in \{-t, t\}^K}  \E_{P_{\bm h, \bm{\overline{\sigma}}}} \lt[\sum_{k = 1}^K p_k h_k (\indicator{h_k > 0} - \widetilde\pi_k(Z))\rt]
        \\
        &= \inf_{\widetilde\pi: \R^K \to [0, 1]^K} \max_{\bm h  \in \{-t, t\}^K} \E_{P_{\bm h, \bm{\overline{\sigma}}}} \lt[\sum_{k = 1}^K p_k |h_k| |\indicator{h_k > 0} - \widetilde\pi_k(Z)|\rt]
        \\
        &= t \inf_{\widetilde\pi: \R^K \to [0, 1]^K} \max_{\bm h  \in \{-t, t\}^K} \E_{P_{\bm h, \bm{\overline{\sigma}}}} \lt[\sum_{k = 1}^K p_k |\indicator{h_k > 0} - \widetilde\pi_k(Z)|\rt]
        \\
        &= t \inf_{\widetilde\pi: \R^K \to [0, 1]^K} \max_{\bm h  \in \{-t, t\}^K} \E_{P_{\bm h, \bm{\overline{\sigma}}}^Z} \lt[\sum_{k = 1}^K p_k |\indicator{h_k > 0} - \widetilde\pi_k(Z)|\rt].
    \end{align*}
where $P_{\bm h, \bm \sigma}^Z = \cN(\bm h, \operatorname{diag}(\sigma_1^2, \ldots, \sigma_K^2))$ is the $Z$-marginal of $P_{\bm h, \bm \sigma}$.

    By Lemma~\ref{lem:weighted_assouad} we obtain
    \begin{align*}
       \inf_{\widetilde\pi: \R^K \times [0, \infty)^K \to [0, 1]^K} \sup_{P_{\bm h, \bm \sigma} \in \mathcal P_K} R_K(P_{\bm h, \bm \sigma}, \widetilde\pi) &\geq \frac{t}{2} \lt(1 - \max_{\substack{\bm h, \bm h^\prime \in \{-t, t\}^K, \\ \rho(\bm h, \bm  h^\prime) = 1}}  \TV(P_{\bm h, \bm{\overline{\sigma}}}^Z, P_{\bm h^\prime, \bm{\overline{\sigma}}}^Z)\rt) \sum_{k=1}^K p_k
       \\
       & \geq \frac{t}{2} \lt(1 - \max_{\substack{\bm h, \bm h^\prime \in \{-t, t\}^K, \\ \rho(\bm h, \bm  h^\prime) = 1}} \TV(P_{\bm h, \bm{\overline{\sigma}}}^Z, P_{\bm h^\prime, \bm{\overline{\sigma}}}^Z)\rt),
    \end{align*}
    where $\rho(\bm h, \bm h^\prime) = \sum_{k=1}^K \indicator{h_k \neq h^\prime_k}$ is the hamming distance and $\TV(\cdot, \cdot)$ the total variation distance between probability measures \citep[Chapter 15.1.3]{wainwright2019high}.
    Observe that when $\rho(\bm h, \bm h^\prime) = 1$, $P_{\bm h, \bm{\overline{\sigma}}}^Z, P_{\bm h^\prime, \bm{\overline{\sigma}}}^Z$ are two $K$-dimensional Gaussian laws that only differ in one component of their mean (by a distance of $2t$) but share same covariance structure of $\overline \sigma^2 I$. By \citet[Equation~(7.40)]{polyanskiy2025information}
    \[
        \TV(P_{\bm h, \bm{\overline{\sigma}}}^Z, P_{\bm h^\prime, \bm{\overline{\sigma}}}^Z) = 2 \Phi \lt(\frac{t}{\overline \sigma}\rt) - 1.
    \]
    Plugging this into the above,
    \begin{align*}
       \inf_{\widetilde\pi: \R^K \times [0, \infty)^K \to [0, 1]^K} \sup_{P_{\bm h, \bm \sigma} \in \mathcal P_K} R_K(P_{\bm h, \bm \sigma}, \widetilde\pi) & \geq  \frac{t}{2} \lt[1 - \lt(2 \Phi \lt(\frac{t}{\overline \sigma}\rt) - 1\rt)\rt]
        \\
        &= t\lt[1 - \Phi \lt(\frac{t}{\overline \sigma}\rt)\rt]
        \\
        &= t \Phi\lt(- \frac{t}{\overline \sigma}\rt).
    \end{align*}
    Since this inequality holds for every $t\in (0, H]$, we deduce
    \begin{align*}
       \inf_{\widetilde\pi: \R^K \times [0, \infty)^K \to [0, 1]^K} \sup_{P_{\bm h, \bm \sigma} \in \mathcal P_K} R_K(P_{\bm h, \bm \sigma}, \widetilde\pi) & \geq \max_{0 \leq t \leq H} t \Phi\lt(- \frac{t}{\overline \sigma}\rt),
    \end{align*}
    which concludes the first part of the proof. 
    
    Next we show that the plug-in policy $\widetilde\pi_k^*(Z, \widehat \sigma) = \indicator{Z_k > 0}$ attains the lower bound. Fix any $P_{\bm h, \bm \sigma} \in \mathcal P_K$. Since the rule ignores $\widehat \sigma$, we have
    \begin{align}\label{eq:scaled_regret}
         R_K(P_{\bm h, \bm \sigma}, \widetilde\pi^*) &= \E_{P_{\bm h, \bm \sigma}}  \lt[\sum_{k = 1}^K p_k |h_k| |\indicator{h_k > 0} - \indicator{Z_k > 0}|\rt]
        \notag
        \\
        &=  \sum_{k = 1}^K p_k |h_k| \lt[\indicator{h_k > 0} P_{\bm h, \bm \sigma}(Z_k \leq 0) + \indicator{h_k < 0} P_{\bm h, \bm \sigma}(Z_k > 0)\rt]
        \notag
        \\
        &= \sum_{k = 1}^K p_k |h_k| \lt[\indicator{h_k > 0} \Phi\lt(- \frac{|h_k|}{\sigma_k}\rt) + \indicator{h_k < 0} \Phi\lt(- \frac{|h_k|}{\sigma_k}\rt)\rt]
        \notag
        \\
        &= \sum_{k = 1}^K p_k  |h_k| \Phi\lt(- \frac{|h_k|}{\sigma_k}\rt)
        \\
        &\leq \sum_{k = 1}^K p_k  |h_k| \Phi\lt(- \frac{|h_k|}{\overline \sigma} \rt)
        \notag
        \\
        &\leq \max_{0 \leq t \leq H} t \Phi\lt(- \frac{t}{\overline \sigma}\rt) \sum_{k = 1}^K p_k 
        \notag
        \\
        &= \max_{0 \leq t \leq H} t \Phi\lt(- \frac{t}{\overline \sigma}\rt).
        \notag
    \end{align}
    Equality is obtained by setting $\bm \sigma = \bm{\overline{\sigma}}$ and any $\bm h$ with
    $|h_k| = t^*_{H, \overline \sigma}$ for all $k$, where $t^*_{H, \overline \sigma} = \argmax_{0 \leq t \leq H} t \Phi\lt(- \frac{t}{\overline \sigma}\rt)$, since the regret depends on $\bm h$ only through $(|h_k|)_{k=1}^K$. This shows that
    \begin{align*}
        \inf_{\widetilde\pi: \R^K \times [0, \infty)^K \to [0, 1]^K} \sup_{P_{\bm h, \bm \sigma} \in \mathcal P_K} R_K(P_{\bm h, \bm \sigma}, \widetilde\pi) = \max_{0 \leq t \leq H} t \Phi\lt(- \frac{t}{\overline \sigma}\rt).
    \end{align*}

    Finally, we prove the Bayes claim. Since for any prior $\Pi \in \mathfrak P$,
    \begin{align}\label{eq:minimax_bayes}
         \inf_{\widetilde\pi: \R^K \times [0, \infty)^K \to [0, 1]^K} \sup_{P_{\bm h, \bm \sigma} \in \mathcal P_K} R_K(P_{\bm h, \bm \sigma}, \widetilde\pi) \geq   \inf_{\widetilde\pi: \R^K \times [0, \infty)^K \to [0, 1]^K} \int R_K(P_{\bm h, \bm \sigma}, \widetilde\pi) \Pi(dP_{\bm h, \bm \sigma}),
    \end{align}
    it is enough to show that the reverse inequality holds for a least-favorable prior $\Pi^* \in \mathfrak P$. Write $t^*=t_{H,\overline\sigma}^*$ and define the prior $\Pi^\ast$ by
    \[
        \Pi^\ast\lt(\{P_{\bm h, \bm{\overline{\sigma}}}\}\rt) = 2^{-K},
    \]
    for $P_{\bm h, \bm{\overline{\sigma}}} \in \cP_K^0$ and $\bm h \in \{-t^*, t^*\}^K$. This prior assigns equal probability to all elements in $\cP^0_K$ with $\bm h \in \{-t^*, t^*\}^K$. Then, following the above calculations and Lemma~\ref{lem:weighted_assouad},
    \begin{align*}
         \inf_{\widetilde\pi: \R^K \times [0, \infty)^K \to [0, 1]^K} \sup_{P_{\bm h, \bm \sigma} \in \mathcal P_K} R_K(P_{\bm h, \bm \sigma}, \widetilde\pi) &\geq   \inf_{\widetilde\pi: \R^K \times [0, \infty)^K \to [0, 1]^K} \int R_K(P_{\bm h, \bm \sigma}, \widetilde\pi) \Pi^*(dP_{\bm h, \bm \sigma})
         \\
         &=  \inf_{\widetilde\pi: \R^K \times [0, \infty)^K \to [0, 1]^K} \frac{1}{2^K} \sum_{\bm h \in \{-t^*, t^*\}^K}  R_K(P_{\bm h, \bm{\overline{\sigma}}}, \widetilde\pi)
         \\
         &=  \inf_{\widetilde\pi: \R^K \to [0, 1]^K} \frac{1}{2^K} \sum_{\bm h \in \{-t^*, t^*\}^K}  R_K(P_{\bm h, \bm{\overline{\sigma}}}, \widetilde\pi)
         \\
         &=  t^* \inf_{\widetilde\pi: \R^K \to [0, 1]^K} \frac{1}{2^K} \sum_{\bm h \in \{-t^*, t^*\}^K}  \E_{P_{\bm h, \bm{\overline{\sigma}}}} \lt[\sum_{k = 1}^K p_k |\indicator{h_k > 0} - \widetilde\pi_k(Z)|\rt]
         \\
         &\geq t^* \Phi\lt(- \frac{t^*}{\overline \sigma}\rt)
         \\
         &= \max_{0 \leq t \leq H} t \Phi\lt(- \frac{t}{\overline \sigma}\rt)
         \\
         &=  \inf_{\widetilde\pi: \R^K \times [0, \infty)^K \to [0, 1]^K} \sup_{P_{\bm h, \bm \sigma} \in \mathcal P_K} R_K(P_{\bm h, \bm \sigma}, \widetilde\pi).
    \end{align*}
    Thus $\Pi^\ast$ is least favorable. By Equation~\eqref{eq:scaled_regret}, the plug-in policy $\widetilde\pi^*$ attains this Bayes risk, which completes the proof.

\subsection{Proof of Theorem~\ref{thm:value}} \label{app:value}

The proof follows from an adaptation of Theorem~3 and Theorem~4 in \citet{cai2011}. To give intuition, establish that our estimator achieves the claimed rate and refer to \citet[Theorem~3]{cai2011} for the corresponding minimax lower bound. In the following, we write $\E_{\bm h, \bm \sigma}$ and $\V_{\bm h, \bm \sigma}$ for expectation and variance under the Gaussian marginal $Z \sim \cN(\bm h, \operatorname{diag}(\sigma_1^2, \ldots, \sigma_K^2))$.

Decompose 
\begin{align*}
    \sqrt{n} \{V_{\bm h}^\ast - \widehat V^\ast_m\} &= \lt[\frac{1}{2} \sum_{k=1}^K p_k h_k + \frac{1}{2} \sum_{k=1}^K p_k |h_k|\rt] - \lt[\frac{1}{2} \sum_{k=1}^K p_k \lt(Z_k + \sum_{j = 0}^m q^\ast_{j, m} \sigma_k^{2j} \Theta_{2j} \lt(\frac{Z_k}{\sigma_k}\rt) \rt)\rt]
    \\
    &= \lt[\frac{1}{2} \sum_{k=1}^K p_k (h_k - Z_k) \rt] + \lt[\frac{1}{2} \sum_{k=1}^K p_k \lt(|h_k| - \sum_{j = 0}^m q^\ast_{j, m} \sigma_k^{2j} \Theta_{2j} \lt(\frac{Z_k}{\sigma_k}\rt) \rt)\rt]
    \\
    &= R_1 + R_2.
\end{align*}
To bound the mean-square error, we study the bias and variance of $R_1$ and $R_2$ respectively. We start with $R_1$. Since $\E_{\bm h, \bm \sigma}[Z_k] = h_k$, $\E_{\bm h, \bm \sigma}[R_1] = 0$. By independence of $Z_1, \ldots, Z_K$,
\begin{align*}
   \V_{\bm h, \bm \sigma}(R_1) = \frac{1}{4} \sum_{k = 1}^K p_k^2 \sigma_k^2 \leq \frac{\overline{\sigma}^2}{4} \sum_{k = 1}^K p_k^2  \lesssim \frac{\overline{\sigma}^2}{4 K},
\end{align*}
as $\max_k p_k = O(1/K)$. Combining bias and variance yields
\begin{align}\label{eq:mean_mse}
    \E_{\bm h, \bm \sigma}[R_1^2] \lesssim \frac{1}{K}.
\end{align}
For the second term $R_2$ we recall that for probabilistic Hermite polynomials and $Z_k \sim \cN(h_k, \sigma_k^2)$
\[
    \E_{\bm h, \bm \sigma}\lt[\sigma_k^{2j} \Theta_{2j} \lt(\frac{Z_k}{\sigma_k}\rt)\rt] = h_k^{2j}.
\]
We use this to determine the bias,
\begin{align*}
    |\E_{\bm h, \bm \sigma}[R_2]| &= \lt|\frac{1}{2} \sum_{k=1}^K p_k \E_{\bm h, \bm \sigma}\lt[|h_k| - \sum_{j = 0}^m q^\ast_{j, m} \sigma_k^{2j} \Theta_{2j} \lt(\frac{Z_k}{\sigma_k}\rt)\rt]\rt| 
    \\
    &= \lt|\frac{1}{2} \sum_{k=1}^K p_k \lt(|h_k| -  \sum_{j = 0}^m q^\ast_{j, m} h_k^{2j}\rt) \rt|
    \\
    &= \lt|\frac{1}{2} \sum_{k=1}^K p_k \lt(|h_k| -  Q^\ast_m(h_k)\rt) \rt|
    \\
    &\leq \frac{1}{2} \delta_m,
\end{align*}
where $\delta_m = \sup_{x \in [-H, H]} ||x| - Q^\ast_m(x)|$ and $Q^\ast_m = \sum_{j=0}^m q^\ast_{j, m} x^{2j}$. Since $Q^\ast_m$ is the best polynomial approximation of $|x|$ on $[-H, H]$ we have \citep{bernstein1913valeur, varga1987conjecture}
\[
    \delta_m \lesssim \frac{H}{m}.
\]
Thus the bias of $R_2$ is uniformly bounded by
\[
    \sup_{\bm h \in [-H, H]^K} |\E_{\bm h, \bm \sigma}[R_2]| \lesssim \frac{H}{m}.
\]
To determine the variance of $R_2$, \citet[Lemma~2 and Section~4.2]{cai2011} bounds the coefficients of $Q^\ast_m$
\[
    |q^\ast_{j, m}| \leq H^{1 - 2j} 2^{3m}.
\]
Hence for $1 \leq j \leq m$ \citet[Lemma~3]{cai2011} implies
\begin{align*}
    \V_{\bm h, \bm \sigma}\lt(q^\ast_{j, m} \sigma_k^{2j} \Theta_{2j} \lt(\frac{Z_k}{\sigma_k}\rt) \rt) &= \lt(q^\ast_{j, m} \sigma_k^{2j}\rt)^2 \V\lt(\Theta_{2j} \lt(\frac{Z_k}{\sigma_k}\rt)\rt)
    \\
    &\leq (q^\ast_{j, m})^2 \sigma_k^{4j} e^{(h_k / \sigma_k)^2} (2j)^{2j}
    \\
    &\leq (H^{1 - 2j} 2^{3m})^2 \sigma_k^{4j}  e^{(h_k / \sigma_k)^2} (2j)^{2j}
    \\
    &\leq H^2  2^{6m} \lt(\frac{\sigma_k}{H}\rt)^{4j} e^{(H / \underline{\sigma})^2} (2j)^{2j}
    \\
        &\leq H^2  2^{6m} \lt(\frac{\sigma_k}{H}\rt)^{4j} e^{(H / \underline{\sigma})^2} 2^{2m} m^{2m}
    \\
    &\leq H^2 e^{(H / \underline{\sigma})^2} 2^{8m} \max\lt\{1, \lt(\frac{\overline{\sigma}}{H}\rt)^{4m}\rt\} m^{2m}
    \\
    &\lesssim A^{m} m^{2m},
\end{align*}
for some large enough $A = A(H, \overline{\sigma}, \underline{\sigma}) > 0$. For $j = 0$ the above term is constant as $\Theta_0(x) \equiv 1$. Using independence across $k$ and the Minkowski’s inequality, 
\begin{align*}
    \V_{\bm h, \bm \sigma}(R_2) &= \frac{1}{4} \sum_{k = 1}^K p_k^2 \V_{\bm h, \bm \sigma}\lt(\sum_{j = 0}^m q^\ast_{j, m} \sigma_k^{2j} \Theta_{2j} \lt(\frac{Z_k}{\sigma_k}\rt) \rt)
    \\
    &\leq \frac{1}{4} \sum_{k = 1}^K p_k^2 \lt(\sum_{j = 0}^m\V_{\bm h, \bm \sigma}\lt(q^\ast_{j, m} \sigma_k^{2j} \Theta_{2j} \lt(\frac{Z_k}{\sigma_k}\rt) \rt)^{1/2}\rt)^2
    \\
    &\lesssim \sum_{k = 1}^K p_k^2 m^2 A^m m^{2m}
    \\
    &\lesssim \frac{1}{K} A^{m} m^{2m}.
\end{align*}
Adding bias and variance yields 
\begin{align}\label{eq:herm_mse}
    \E_{\bm h, \bm \sigma}[R_2^2] \lesssim \frac{H^2}{m^2} + \frac{A^m m^{2m}}{K}.
\end{align}
Combining Equations~\eqref{eq:mean_mse} and~\eqref{eq:herm_mse}, for any $h \in [-H, H]^K$,
\begin{align*}
   \E_{\bm h, \bm \sigma}\lt[\lt(\sqrt{n} \{V^\ast_{\bm h} - \widehat V^\ast_m\}\rt)^2\rt] &= \E_{\bm h, \bm \sigma}[(R_1 + R_2)^2]
   \\
   &\leq 2 \E_{\bm h, \bm \sigma}[R_1^2] + 2 \E_{\bm h, \bm \sigma}[R_2^2]
   \\
   &\lesssim \frac{1}{K} + \frac{H^2}{m^2} + \frac{A^m m^{2m}}{K}
   \\
   &\lesssim_{H, \overline{\sigma}, \underline{\sigma}} \lt(\frac{\log(\log(K))}{\log(K)}\rt)^2,
\end{align*}
where we set $m = m_K = \lt \lfloor \frac{1}{4} \log(K) / (\log(\log(K))) \rt \rfloor$ in the last line to balance the terms. This proves the claimed rate.

\subsection{Proof of Proposition~\ref{pro:smoothing_tradeoff}} \label{app:smoothing_tradeoff}

Under Assumptions~\ref{ass:gaussian_model}--\ref{ass:weak_signal}, $\widehat g_w(x) \sim \cN(m_w(x), s_w(x)^2)$. Its induced statistical decision rule is given by $\widetilde\pi_{w, k}(Z) = \indicator{\widehat g_w(x_k) > 0}$ and incurs a regret of
\begin{align}\label{eq:sign_regret}
    R_K(P_{g, \bm{\sigma}}, \widetilde\pi_w) &= \sum_{k= 1}^K p_k |g(x_k)| |\E_{g, \bm{\sigma}}[\indicator{\widehat g_w(x_k) > 0}] - \indicator{g(x_k) > 0}|
    \notag
    \\
    &= \sum_{k= 1}^K p_k |g(x_k)| P_{g, \bm{\sigma}}(\indicator{\widehat g_w(x_k) > 0} \neq \indicator{g(x_k) > 0})
    \notag
    \\
    &= \sum_{k = 1}^K p_k |g(x_k)| P_{g, \bm{\sigma}}(\widehat g_w(x_k) g(x_k) \leq  0)
    \notag
    \\
    &= \sum_{k= 1}^K p_k |g(x_k)| \Phi\lt(- \frac{m_w(x_k) \operatorname{sign}(g(x_k))}{s_w(x_k)} \rt).
\end{align}
We analyze the regret coordinatewise.  If $g(x_k) m_w(x_k) \leq 0$ then 
\begin{align*}
    |g(x_k)| \Phi\lt(-\frac{m_w(x_k) \operatorname{sign}(g(x_k))}{s_w(x_k)} \rt) \leq |g(x_k)| \leq |g(x_k) - m_w(x_k)|.
\end{align*}
Now suppose that $g(x_k)m_w(x_k)>0$. Then
$m_w(x_k)\operatorname{sign}(g(x_k))=|m_w(x_k)|$, and hence
\begin{align*}
    |g(x_k)| \Phi\lt( -\frac{|m_w(x_k)|}{s_w(x_k)}\rt) &\leq
    \Big(
        |g(x_k)-m_w(x_k)|
        +
        |m_w(x_k)|
    \Big)
    \Phi\lt(
        -\frac{|m_w(x_k)|}{s_w(x_k)}
    \rt)
    \\
    &\leq
    |g(x_k)-m_w(x_k)|
    + s_w(x_k) \frac{|m_w(x_k)|}{s_w(x_k)}
    \Phi\lt( -\frac{|m_w(x_k)|}{s_w(x_k)} \rt)
    \\
    &\leq
    |g(x_k)-m_w(x_k)| + \gamma^\ast s_w(x_k),
\end{align*}
where $\gamma^\ast := \sup_{t\geq0}t\Phi(-t)$. Combining both cases gives
\begin{align*}
     R_K(P_{g, \bm{\sigma}}, \widetilde\pi_w) &= \sum_{k= 1}^K p_k |g(x_k)| \Phi\lt(- \frac{m_w(x_k) \operatorname{sign}(g(x_k))}{s_w(x_k)} \rt)
     \\
     &\leq \sum_{k= 1}^K p_k \Big(|g(x_k)-m_w(x_k)| + \gamma^\ast s_w(x_k)\Big)
     \\
     &= \sA_g(w) + \gamma^\ast \sE_{\bm{\sigma}}(w).
\end{align*}
Finally, since $\|g\|_\infty \leq H$ and $\sum_{k=1}^Kp_k=1$, we have trivially, 
\[
    R_K(P_{g, \bm{\sigma}},\widetilde\pi_w) \leq  \sum_{k=1}^Kp_k|g(x_k)| \leq H.
\]

We now show that the bound is sharp over the class of linear smoothers, both for the approximation term and for the estimation term. To establish sharpness of the coefficient on the estimation error, let $ t^\ast\in\argmax_{t\ge0} t\Phi(-t)$ and suppose that
\[
    g(x)=\frac{t^\ast\sigma}{\sqrt K},
\]
$\sigma_k=\sigma$ and $p_k=\frac{1}{K}$ for all $k \in [K]$. Trivially, $g \in \cH_d(\beta, L, H)$.
Consider the uniform linear smoother with $w_j(x_k)=1/K$ for all $j,k \in [K]$, which sum to one. Then for all $k \in [K]$, $m_w(x_k) = g(x_k)$ and hence $\sA_g(w)=0$. Furthermore, $s_w(x_k) = \frac{\sigma}{\sqrt{K}}$ for all $k \in [K]$, which implies $\sE_{\bm\sigma}(w)=\frac{\sigma}{\sqrt K}$. Thus,
\begin{align*}
    R_K(P_{g, \bm{\sigma}},\widetilde\pi_w) &= \frac{t^\ast \sigma}{\sqrt{K}} \Phi(-t^\ast)
    \\
    &= \gamma^\ast \frac{\sigma}{\sqrt{K}}
    \\
    &= \gamma^\ast \sE_{\bm\sigma}(w)
\end{align*}
and $\gamma^\ast$ is sharp.

Next, we show that the approximation coefficient (of $1$) is sharp. Let $K$ be even and define $\varepsilon_K=\frac{H}{K^{1/4}}$. Let $\Delta=\lt(\frac{2H}{L}\rt)^{1/\beta}$ and define the signal by
\[
    g_K(x)
    =
    \begin{cases}
        H, & \text{if } x^{(1)}\leq 0,\\
        -\varepsilon_K, & \text{if } x^{(1)}\geq \Delta,
    \end{cases}
\]
with appropriate $\beta$-H\"older smoothing in between, where $x^{(1)}$ denotes the first coordinate of $x$. For example, one may take $g_K(x) = -\varepsilon_K + (H+\varepsilon_K) f(x^{(1)})$ where $f(u) := \lt[1 - \lt(\frac{u^+}{\Delta}\rt)^\beta\rt]^+$. Hence $g_K \in \cH_d(\beta, L, H)$.

Let $p_k=1/K$ and $\sigma_k=\sigma$ for all $k\in[K]$. Choose task covariates $x_1,\ldots,x_{K/2}$ such that $g_K(x_k)=H$, and choose $x_{K/2+1},\ldots,x_K$ such that $g_K(x_k)=-\varepsilon_K$. Define the smoothing weights
\[
    w_j(x)  = \frac{2}{K}\indicator{g_K(x_j)\leq 0}.
\]
These weights are deterministic, nonnegative, and sum to one by construction. They induce the smoother
\[
    \widehat g_w(x) =  \frac{2}{K} \sum_{j:\,g_K(x_j)\le0} Z_j,
\]
which averages only over the negative-signal tasks.

Its mean is given by
\[
    m_w(x_k) = \frac{2}{K} \sum_{j: g(x_j) \leq 0} - \epsilon_K = - \epsilon_K,
\]
and variance by
\begin{align*}
    s_w(x_k)^2 &= \lt(\frac{2}{K}\rt)^2 \sum_{j: g(x_j) \leq 0} \sigma^2
    \\
    &= \frac{4 \sigma^2}{K^2} \frac{K}{2}
    \\
    &= \frac{2 \sigma^2}{K}.
\end{align*}
Consequently, by Equation~\eqref{eq:sign_regret}
\begin{align*}
    R_K(P_{g, \bm{\sigma}}, \widetilde\pi_w) &= \sum_{k= 1}^K p_k |g(x_k)| \Phi\lt(- \frac{m_w(x_k) \operatorname{sign}(g(x_k))}{s_w(x_k)} \rt)
    \\
    &= \frac{1}{K} \sum_{k = 1}^{K/2} H \Phi\lt(\frac{\epsilon_K}{\sigma \sqrt{2 / K}}\rt) + \frac{1}{K} \sum_{k = K/2 + 1}^{K} \epsilon_K \Phi\lt(\frac{- \epsilon_K}{\sigma \sqrt{2 / K}}\rt)
    \\
    &= \frac{1}{K} \sum_{k = 1}^{K/2} H \Phi\lt(\frac{H K^{1/4}}{\sigma \sqrt{2}}\rt) + \frac{1}{K} \sum_{k = K/2 + 1}^{K} \epsilon_K \Phi\lt(\frac{-H K^{1/4}}{\sigma \sqrt{2}}\rt)
    \\
    &= \frac{H}{2} \Phi\lt(\frac{H K^{1/4}}{\sigma \sqrt{2}}\rt) + \frac{\epsilon_K}{2} \Phi\lt(\frac{-H K^{1/4}}{\sigma \sqrt{2}}\rt).
\end{align*}
Thus, $R_K(P_{g, \bm{\sigma}}, \widetilde\pi_w) \to \frac{H}{2}$ as $K \to \infty$. Since mistakes are made only on tasks with positive signal, the approximation error is
\begin{align*}
    \sA_{g_K}(w) = \frac{1}{2} |H - (- \epsilon_K)| \to \frac{H}{2}.
\end{align*}
Since all tasks are equally weighted and have the same standard error, the estimation error is given by
\[
    \sE_{\bm{\sigma}}(w) = \sigma \sqrt{\frac{2}{K}} \to 0.
\]
Hence, in the limit $R_K(P_{g, \bm{\sigma}}, \widetilde\pi_w) -  \sA_{g_K}(w) \to 0$, which proves that the approximation constant is sharp.

\subsection{Proof of Theorem~\ref{thm:regret_linear_smoother}}\label{app:regret_linear_smoother}

First, by Assumption~\ref{ass:quasi_uniform_design}, for every $j \in [K]$, we have
\begin{align*}
    p_j \leq Q_K(B(x_j, a_0K^{-1/d})) \leq \frac{C a_0^d}{K},
\end{align*}
and consequently, 
\begin{align}\label{eq:max_p}
    \max_{j \in [K]} p_j \leq \frac{C a_0^d}{K}.
\end{align}

We now prove the theorem for partition pooling, then for kernel smoothing and finally for series smoothing estimators separately.

\paragraph{Partition Pooling.}
Fix $k \in [K]$ and let $A_k=A_r(x_k)$ be the cell of $\cA_r$ that contains $x_k$. Then $w_{r, j}^{\mathrm P}(x_k) = \frac{p_j \indicator{x_j\in A_k}}{Q_K(A_k)}$ where $Q_K(A_k) = \sum_{j: x_j \in A_k} p_j$.
By Proposition~\ref{pro:smoothing_tradeoff}, the approximation error is given by
\begin{align*}
    \sA_g(w_r^{\mathrm P}) &= \sum_{k = 1}^K p_k |g(x_k) - m_{w_r^{\mathrm P}}(x_k)|
    \\
    &= \sum_{k = 1}^K p_k \lt|g(x_k) - \sum_{j: x_j \in A_k} \frac{p_j}{Q_K(A_k)} g(x_j)\rt|
    \\
    &\leq \sum_{k = 1}^K p_k \lt|\sum_{j: x_j \in A_k} \frac{p_j}{Q_K(A_k)} g(x_k) - \sum_{j: x_j \in A_k} \frac{p_j}{Q_K(A_k)} g(x_j)\rt|
    \\
    &\leq \sum_{k = 1}^K p_k \sum_{j: x_j \in A_k} \frac{p_j}{Q_K(A_k)} \lt| g(x_k) - g(x_j)\rt|
    \\
    &\leq \sum_{k = 1}^K p_k \sum_{j: x_j \in A_k} \frac{p_j}{Q_K(A_k)} L r^\beta
    \\
    &\leq \sum_{k = 1}^K p_k L r^\beta
    \\
    &\leq L r^\beta.
\end{align*}
Likewise, for the estimation error, by Assumption~\ref{ass:gaussian_model},
\begin{align*}
    \sE_{\bm{\sigma}}(w_r^{\mathrm P}) &= \sum_{k = 1}^K p_k s_{w_r^{\mathrm P}}(x_k)    
    \\
    &= \sum_{k = 1}^K p_k \sqrt{\sum_{j=1}^K w_{r, j}^{\mathrm P}(x_k)^2\sigma_j^2}
    \\
    &\leq \overline{\sigma} \sum_{k = 1}^K p_k \sqrt{\sum_{j=1}^K w_{r, j}^{\mathrm P}(x_k)^2}
    \\
    &= \overline{\sigma} \sum_{k = 1}^K p_k \sqrt{\sum_{j: x_j \in A_k} \frac{ p_j^2}{Q_K(A_k)^2}}
    \\
    &\leq \overline{\sigma} \sum_{k = 1}^K p_k \sqrt{ \frac{ \max_{j \in [K]} p_j \sum_{j: x_j \in A_k} p_j}{Q_K(A_k)^2}}
    \\
    &= \overline{\sigma} \sum_{k = 1}^K p_k  \sqrt{\frac{ \max_{j \in [K]} p_j}{Q_K(A_k)}}
    \\
    &\leq \overline{\sigma} \sum_{k = 1}^K p_k  \sqrt{\frac{C a_0^d}{c_{\mathrm P} K r^d}}
    \\
    &= \sqrt{\frac{C a_0^d}{c_{\mathrm P}}} \frac{\overline{\sigma} }{\sqrt{K r^d}}
    \\
    &\lesssim \frac{\overline{\sigma} }{\sqrt{K r^d}}
\end{align*}
where we used Equation~\eqref{eq:max_p} and that $\min_{A \in \cA_r} Q_K(A) \geq c_{\mathrm P} r^d$ by assumption. This proves the partition pooling claim.

\paragraph{Kernel Smoothing.}
We consider the kernel smoothing claim next. Because the kernel is supported on the unit ball, when $w_{r,j}^{\mathrm K}(x_k)>0$ then $\|x_j-x_k\|_2\le r$. Hence, 
\begin{align*}
    \sA_g(w_r^{\mathrm K}) &= \sum_{k = 1}^K p_k |g(x_k) - m_{w_r^{\mathrm K}}(x_k)|
    \\
    &=\sum_{k = 1}^K p_k |g(x_k) - \sum_{j = 1}^K w_{r, j}^{\mathrm K}(x_k) g(x_j)|
    \\
    &=\sum_{k = 1}^K p_k \lt|\sum_{j = 1}^K w_{r, j}^{\mathrm K}(x_k) g(x_k) - \sum_{j = 1}^K w_{r, j}^{\mathrm K}(x_k) g(x_j)\rt|
    \\
    &\leq \sum_{k = 1}^K p_k \sum_{j = 1}^K w_{r, j}^{\mathrm K}(x_k) | g(x_k) - g(x_j)|
    \\
    &\leq \sum_{k = 1}^K p_k \sum_{j = 1}^K w_{r, j}^{\mathrm K}(x_k) L r^\beta
    \\
    &\leq L r^\beta.
\end{align*}

Next, we upper bound the estimation error. We start by bounding the squared sum of weights,
\begin{align}\label{eq:kernel_weights_square}
    \sum_{j=1}^K w_{r, j}^{\mathrm K}(x_k)^2 &\leq  \sum_{j=1}^K \frac{\sK\lt((x_j-x_k)/r\rt)^2}
    {\lt(\sum_{\ell=1}^K\sK\lt((x_\ell-x_k)/r\rt)\rt)^2}
    \notag
    \\
    &\leq \|\sK\|_\infty \sum_{j=1}^K \frac{\sK\lt((x_j-x_k)/r\rt)}
    {\lt(\sum_{\ell=1}^K\sK\lt((x_\ell-x_k)/r\rt)\rt)^2}
    \notag
    \\
    &= \frac{\|\sK\|_\infty}{\sum_{\ell=1}^K\sK\lt((x_\ell-x_k)/r\rt)}
\end{align}
Convexity also gives $\sum_{j=1}^K w_{r, j}^{\mathrm K}(x_k)^2 \leq 1$. 

Because the kernel is continuous at zero and positive at zero, there are constants $\rho\in(0,1]$ and $\kappa>0$ such that $ \sK(u)\ge\kappa$ for all $\|u\|_2\le\rho$. This induces two regimes. If $\rho r\geq a_0K^{-1/d}$, then
\begin{align*}
    \sum_{\ell=1}^K
    \sK\lt(\frac{x_\ell-x_k}{r}\rt)
    &\geq \kappa \sum_{\ell=1}^K \indicator{\|x_\ell-x_k\|_2\leq\rho r}
    \\
    &\geq \kappa \sum_{\ell=1}^K \frac{p_\ell}{\max_{\ell\in[K]}p_\ell}\indicator{\|x_\ell-x_k\|_2\leq\rho r}
    \\
    &=
    \kappa
    \frac{Q_K(B(x_k,\rho r))}{\max_{\ell\in[K]}p_\ell}
    \\
    &\geq
    \kappa
    \frac{c\rho^dr^d}{Ca_0^d/K}
    \\
    &=
    \frac{\kappa c\rho^d}{Ca_0^d}Kr^d,
\end{align*}
where we used
$\max_{\ell\in[K]}p_\ell\leq Ca_0^d/K$ by Equation~\eqref{eq:max_p} and $Q_K(B(x_k,\rho r)) \geq c \rho^d r^d$ by Assumption~\ref{ass:quasi_uniform_design}. Consequently, if $\rho r\geq a_0K^{-1/d}$, Equation~\eqref{eq:kernel_weights_square} gives
\[
    \sum_{j=1}^K
    w_{r,j}^{\mathrm K}(x_k)^2
    \leq
    \frac{C a_0^d\|\sK\|_\infty}
         {\kappa c\rho^d}
    \frac{1}{Kr^d}.
\]

If instead $\rho r<a_0K^{-1/d}$, then $Kr^d<\frac{a_0^d}{\rho^d}$.
Therefore, using
$\sum_{j=1}^K w_{r,j}^{\mathrm K}(x_k)^2\leq1$,
\[
    \sum_{j=1}^K
    w_{r,j}^{\mathrm K}(x_k)^2
    \leq
    1
    \leq
    \frac{a_0^d}{\rho^d}
    \frac{1}{Kr^d}.
\]

Combining the two regimes gives
\[
    \sum_{j=1}^K
    w_{r,j}^{\mathrm K}(x_k)^2
    \leq
    \frac{a_0^d}{\rho^d}
    \max\lt\{
        1,
        \frac{C\|\sK\|_\infty}{\kappa c}
    \rt\}
    \frac{1}{Kr^d}.
\]
It follows that
\begin{align*}
    \sE_{\bm\sigma}(w_r^{\mathrm K})
    &=
    \sum_{k=1}^Kp_k s_{w_r^{\mathrm K}}(x_k)
    \\
    &= \sum_{k=1}^Kp_k \sqrt{\sum_{j=1}^K w_{r,j}^{\mathrm K}(x_k)^2\sigma_j^2}
    \\
    &\leq \overline{\sigma} \sum_{k=1}^Kp_k \sqrt{\sum_{j=1}^K w_{r,j}^{\mathrm K}(x_k)^2}
    \\
    &\leq
    \frac{a_0^{d/2}}{\rho^{d/2}}
    \max\lt\{
        1,
        \frac{C\|\sK\|_\infty}{\kappa c}
    \rt\}^{1/2}
    \frac{\overline\sigma}{\sqrt{Kr^d}} \sum_{k = 1}^K p_k
    \\
     &\leq
    \frac{a_0^{d/2}}{\rho^{d/2}}
    \max\lt\{
        1,
        \frac{C\|\sK\|_\infty}{\kappa c}
    \rt\}^{1/2}
    \frac{\overline\sigma}{\sqrt{Kr^d}}
    \\
    &\lesssim
    \frac{\overline\sigma}{\sqrt{Kr^d}}.
\end{align*}
This proves the bounds on the kernel smoothing estimator.

\paragraph{Series Regression.}

We start with the approximation error. Fix the a potential signal $g \in \cH_d(\beta, L, H)$. We construct a signal close to $g$that lies within the span of the basis functions.

For each $\nu\in[q_r]$ choose $\xi_\nu\in\operatorname{supp}(b_{r,\nu}) \cap \cX$ and define the spline interior signal as
\begin{align}\label{eq:spline_signal}
    g_r^\circ(x) = \sum_{\nu=1}^{q_r} \theta_\nu^\circ b_{r,\nu}(x),
\end{align}
where $\theta_\nu^\circ = g(\xi_\nu)$ are the spline coefficients.
By assumption, for each $x \in \cX$ exists $\nu$ such that $b_{r,\nu}(x) > 0$. Furthermore, if $x \in \operatorname{supp}(b_{r,\nu})$ then $\|x - \xi_\nu\|_2 \leq C_{\mathrm S}r$. Since $\sum_{\nu=1}^{q_r}b_{r,\nu}(x)=1$, for all $x \in \cX$
\begin{align}\label{eq:sup_signal}
    |g(x)-g_r^\circ(x)| &= \lt|g(x) -  \sum_{\nu=1}^{q_r} \theta_\nu^\circ b_{r,\nu}(x)\rt|\
    \notag
    \\
    &=\lt|g(x) -  \sum_{\nu=1}^{q_r} g(\xi_\nu)b_{r,\nu}(x)\rt|
    \notag
    \\
    &= \lt|\sum_{\nu=1}^{q_r} b_{r,\nu}(x) g(x) -  \sum_{\nu=1}^{q_r} g(\xi_\nu)b_{r,\nu}(x)\rt|
    \notag
    \\
    &\leq \sum_{\nu=1}^{q_r} b_{r,\nu}(x) \lt|g(x) -  g(\xi_\nu)\rt|
    \notag
    \\
    &\leq L \sum_{\nu=1}^{q_r} b_{r,\nu}(x) \lt\|x -  \xi_\nu \rt\|_2^\beta
    \notag
    \\
    &\leq L C_{\mathrm S}^\beta r^\beta \sum_{\nu=1}^{q_r} b_{r,\nu}(x) 
    \notag
     \\
    &\leq L C_{\mathrm S}^\beta r^\beta,
\end{align}
where we used that $g \in \cH_d(\beta, L, H)$. 

Next, because  $a_k=p_k$, we have $m_{w_r^{\mathrm{S}}}(x) = b_r(x)^\top \theta_{r,g}^\ast$ where
\begin{align*}
    \theta_{r,g}^\ast &=  \E[\widehat\theta_{r,\bm p}] = G_{r,\bm p}^{-1} \sum_{k = 1}^K p_k b_r(x_k) g(x_k).
\end{align*}
Define $Q_r(\theta) = \sum_{k=1}^K p_k \lt(g(x_k)-b_r(x_k)^\top \theta\rt)^2$. Observe that $\theta_{r,g}^\ast$ is the unique minimizer of $Q_r$, hence $Q_r(\theta_{r,g}^\ast) \leq Q_r(\theta^\circ)$, where $\theta^\circ$ is the coefficient vector of $g_r^\circ$ in Equation~\eqref{eq:spline_signal}.
Thus
\begin{align*}
    \sum_{k=1}^K p_k \lt(g(x_k)-m_{w_r^{\mathrm{S}}}(x_k) \rt)^2 &= \sum_{k=1}^K p_k \lt(g(x_k)- b_r(x_k)^\top \theta_{r,g}^\ast\rt)^2
    \\
    &=  Q_r(\theta_{r,g}^\ast)
    \\
    &\leq Q_r(\theta^\circ)
    \\
    &= \sum_{k=1}^K p_k \lt(g(x_k)- b_r(x_k)^\top \theta^\circ\rt)^2
    \\
    &= \sum_{k=1}^K p_k \lt(g(x_k)- g_r^\circ(x_k)\rt)^2.
\end{align*}
Finally, by Jensen's inequality, $\sum_{k=1}^Kp_k=1$ and Equation~\eqref{eq:sup_signal},
\begin{align*}
    \sA_g(w_r^{\mathrm S})
    &=
    \sum_{k=1}^Kp_k
    |g(x_k)-m_{w_r^{\mathrm S}}(x_k)|
    \\
    &\leq
    \lt[
        \sum_{k=1}^Kp_k
        (g(x_k)-m_{w_r^{\mathrm S}}(x_k))^2
    \rt]^{1/2}
    \\
    &\leq
    \lt[
        \sum_{k=1}^Kp_k
        (g(x_k) - g_r^\circ(x_k))^2
    \rt]^{1/2}
    \\
    &\leq
    \sup_{x\in\cX}|g(x)-g_r^\circ(x)|
    \\
    &\leq  C_{\mathrm S} L r^\beta
    \\
    &\lesssim L r^\beta,
\end{align*}
which concludes the claim on the approximation error.

For the estimation error, the weights in
Definition~\ref{def:series_regression} give, for every $k\in[K]$,
\begin{align*}
    s_{w_r^{\mathrm S}}(x_k)^2
    &=
    \sum_{j=1}^K \sigma_j^2p_j^2 \lt( b_r(x_k)^\top G_{r,\bm p}^{-1}b_r(x_j) \rt)^2
    \\
    &\leq \overline\sigma^2 \sum_{j=1}^K p_j^2 \lt( b_r(x_k)^\top G_{r,\bm p}^{-1}b_r(x_j) \rt)^2
    \\
    &\leq \overline\sigma^2 \max_{j\in[K]}p_j \sum_{j=1}^K p_j \lt( b_r(x_k)^\top G_{r,\bm p}^{-1}b_r(x_j) \rt)^2
    \\
    &=
    \overline\sigma^2 \max_{j\in[K]}p_j b_r(x_k)^\top G_{r,\bm p}^{-1}b_r(x_k).
\end{align*}
The final equality follows from the definition of $G_{r,\bm p}$ and the fact that it is positive definite,
\begin{align*}
    \sum_{j=1}^K p_j \lt( b_r(x_k)^\top G_{r,\bm p}^{-1}b_r(x_j) \rt)^2 &= \sum_{j=1}^K p_j \lt( b_r(x_k)^\top G_{r,\bm p}^{-1}b_r(x_j) \rt) \lt( b_r(x_k)^\top G_{r,\bm p}^{-1}b_r(x_j) \rt)
    \\
    &= \sum_{j=1}^K p_j \lt( b_r(x_k)^\top G_{r,\bm p}^{-1}b_r(x_j) \rt) \lt( b_r(x_j)^\top G_{r,\bm p}^{-1}b_r(x_k) \rt)
    \\
    &=  b_r(x_k)^\top G_{r,\bm p}^{-1} \lt(\sum_{j=1}^K p_j  b_r(x_j) b_r(x_j)^\top \rt) G_{r,\bm p}^{-1}b_r(x_k) 
    \\
    &=  b_r(x_k)^\top G_{r,\bm p}^{-1} G_{r,\bm p} G_{r,\bm p}^{-1}b_r(x_k) 
    \\
    &=  b_r(x_k)^\top G_{r,\bm p}^{-1}b_r(x_k).
\end{align*}

Consequently, by Jensen inequality and the definition of $G_{r,\bm p}$,
\begin{align*}
    \sE_{\bm\sigma}(w_r^{\mathrm S}) &= \sum_{k = 1}^K p_k  s_{w_r^{\mathrm S}}(x_k)
    \\
    &\leq \lt\{\sum_{k = 1}^K p_k  s_{w_r^{\mathrm S}}(x_k)^2 \rt\}^{1/2}
    \\
    &\leq \overline\sigma \sqrt{\max_{j\in[K]}p_j} \lt\{\sum_{k = 1}^K p_k  b_r(x_k)^\top G_{r,\bm p}^{-1}b_r(x_k) \rt\}^{1/2}
    \\
    &= \overline\sigma \sqrt{\max_{j\in[K]}p_j} \lt\{\operatorname{tr}\lt(\sum_{k = 1}^K p_k  b_r(x_k)^\top G_{r,\bm p}^{-1}b_r(x_k) \rt)\rt\}^{1/2}
    \\
    &= \overline\sigma \sqrt{\max_{j\in[K]}p_j} \lt\{\operatorname{tr}\lt( G_{r,\bm p}^{-1}\sum_{k = 1}^K p_k  b_r(x_k) b_r(x_k)^\top \rt)\rt\}^{1/2}
    \\
    &= \overline\sigma \sqrt{\max_{j\in[K]}p_j} \lt\{\operatorname{tr} \lt(G_{r,\bm p}^{-1} G_{r,\bm p} \rt)\rt\}^{1/2}
    \\
    &= \overline\sigma \sqrt{\max_{j\in[K]}p_j} \sqrt{q_r}
    \\
    &\lesssim \overline{\sigma} \sqrt{\frac{C a_0^d}{K r^d}}
    \\
    &\lesssim
    \frac{\overline\sigma}{\sqrt{Kr^d}},
\end{align*}
where the last inequality follows from Equation~\eqref{eq:max_p} and
$q_r\asymp r^{-d}$.

\subsection{Proof of Theorem~\ref{thm:minimax_lower_bound}}\label{app:minimax_lower_bound}

The proof is based on a weighted version of Assouad's lemma
\citep[Theorem~2.12]{tsybakov2009introduction} and adapts techniques from \citet{audibert2007fast}. We construct a family of potential signal functions on the hypercube, whose elements are statistically
difficult to distinguish, while errors in identifying the right sign of each
coordinate incur nonvanishing welfare loss.

We start by constructing a family of signal functions.
Fix $r\in[2a_0K^{-1/d},r_0/2]$. For $x\in\R^d$, define
$\phi_0(x):=(1-\|x\|_2^2)_+^\beta$ and set
$\phi_{z,r}(x):=\phi_0\lt(\frac{x-z}{r}\rt)$. By
Lemma~\ref{lem:low_norm_packing}, there exists a set
$\cZ_r=\{z_1,\ldots,z_M\}\subseteq\{x_k:k\in[K]\}$ such that $\|z_i - z_j\|_2 > 2r$ for $i \neq j, i, j \in [M]$, $M\geq\frac{1}{C2^{d+1}r^d}$, for every $j\in[M]$, 
\[
\sum_{k=1}^K\phi_{z_j,r}(x_k)^2 \leq 2^{d+1}CKr^d,
\]
and 
\[
\sum_{j=1}^M Q_K\lt(B(z_j,r/2)\rt) \geq \frac{c}{C2^{2d+1}},
\]
where $0 < c < C < \infty$ are the constants from Assumption~\ref{ass:quasi_uniform_design}.

Let $\Omega:=\{0,1\}^M$ and define $a := \min\lt\{
        H,
        \frac{Lr^\beta}{2^{\beta+1}},
        \frac{\overline\sigma}{2\sqrt{2^{d+1} C Kr^d}}
    \rt\}.$
For every $\omega=(\omega_1,\ldots,\omega_M)\in\Omega$, set
\[
    g_\omega(x)
    :=
    a\sum_{j=1}^M(2\omega_j-1)\phi_{z_j,r}(x),
\]
which will be our candidate signals, for which we need to show that $g_\omega\in\cH_d(\beta,L,H)$.

We verify that $\|g_\omega\|_\infty\leq a\leq H$.
Since $\operatorname{supp}(\phi_0)=B(0,1)$, we have, for every $j\in[M]$,
$\operatorname{supp}(\phi_{z_j,r})=B(z_j,r)$. By
Lemma~\ref{lem:low_norm_packing}, $\|z_i - z_j\|_2 > 2r$ for all $i \neq j$, which implies that the closed balls $\{B(z_j,r):j\in[M]\}$, and therefore the supports of
$\{\phi_{z_j,r}:j\in[M]\}$, are pairwise disjoint. Thus at most one summand
in $g_\omega(x)$ is nonzero for every $x\in\R^d$. Since
$0\leq\phi_{z_j,r}\leq1$, it follows that
$\|g_\omega\|_\infty\leq a\leq H$.

Next, we verify that $g_\omega$ is $\beta$-H\"older with constant $L$. By Lemma~\ref{lem:hoelder_regularity}
\begin{align*}
    |g_\omega(x)-g_\omega(y)|
    &\leq
    a\sum_{j=1}^M
    |\phi_{z_j,r}(x)-\phi_{z_j,r}(y)|
    \\
     &=
    a\sum_{j=1}^M
    |\phi_{z_j,r}(x)-\phi_{z_j,r}(y)|  \indicator{x\in B(z_j,r)\text{ or }y\in B(z_j,r)}
    \\
    &\leq
    \frac{2^\beta a}{r^\beta}\|x-y\|_2^\beta
    \sum_{j=1}^M
    \indicator{x\in B(z_j,r)\text{ or }y\in B(z_j,r)}
    \\
    &\leq
    \frac{2^{\beta+1}a}{r^\beta}\|x-y\|_2^\beta
    \\
    &\leq
    L\|x-y\|_2^\beta,
\end{align*}
where the penultimate inequality follows because each of $x$ and $y$
belongs to at most one ball.
Together with $\|g_\omega\|_\infty\leq H$, this proves that
$g_\omega\in\cH_d(\beta,L,H)$.

Next, we bound the total variation distance between the data distributions
induced by neighboring hypotheses. We use a similar construction as in the proof of Theorem~\ref{thm:minimax_regret}. Fix a probability law $Q_0$ on
$[0, \infty)^K$ and let $\cP_K^0 \subseteq \cP_K$ consist of the laws
under which $\widehat \sigma \sim Q_0$ and $\widehat \sigma \indep Z$.
Write $P_{g,\bm\sigma}^Z$ for the $Z$-marginal of $P_{g,\bm\sigma}$. By the same argument as in the proof of Theorem~\ref{thm:minimax_regret}, on the submodel $\cP_K^0$, it suffices to prove the lower bound for rules that do not take $\widehat \sigma$ as an input.

For $j\in[M]$, let $\omega^{(j)}$ denote $\omega$ with
coordinate $j$ flipped, i.e., $\omega^{(j)}_{i} = \omega_{i}$ for all $i \neq j$ and $\omega^{(j)}_{j} = 1 - \omega_{j}$.  Let
$P_\omega := P_{g_\omega,\bm{\overline{\sigma}}}^Z$, where
$P_{g_\omega,\bm{\overline{\sigma}}}\in\cP_K^0$ is the corresponding joint law. Then $P_\omega$ and $P_{\omega^{(j)}}$ are
$K$-dimensional Gaussian distributions centered around
$\{g_\omega(x_k)\}_{k=1}^K$ and
$\{g_{\omega^{(j)}}(x_k)\}_{k=1}^K$, respectively, with shared covariance
matrix $\overline\sigma^2I_K$. Since
$g_\omega,g_{\omega^{(j)}}\in\cH_d(\beta,L,H)$ and
$\bm{\overline{\sigma}}\in
[\underline\sigma,\overline\sigma]^K$, $P_\omega$ and
$P_{\omega^{(j)}}$ are potential data-generating distributions of
$Z=(Z_1,\ldots,Z_K)$.

By \citet{barsov1987estimates,devroye2018total},
\begin{align*}
    d_{\mathrm{TV}}(P_\omega,P_{\omega^{(j)}})
    =
    2\Phi\lt(
        \frac{
            \sqrt{\sum_{k=1}^K
            (g_\omega(x_k)-g_{\omega^{(j)}}(x_k))^2}
        }{2\overline\sigma}
    \rt)-1.
\end{align*}
Since $\omega$ and $\omega^{(j)}$ differ only in coordinate $j$,
\begin{align*}
    \frac{1}{2\overline\sigma}
    \sqrt{\sum_{k=1}^K
    (g_\omega(x_k)-g_{\omega^{(j)}}(x_k))^2}
    &=
    \frac{a}{\overline\sigma}
    \sqrt{\sum_{k=1}^K\phi_{z_j,r}(x_k)^2}
    \\
    &\leq
    \frac{a\sqrt{2^{d+1}C Kr^d}}{\overline\sigma}
    \\
    &\leq
    \frac12,
\end{align*}
where the first inequality follows from
Lemma~\ref{lem:low_norm_packing} and the second from the definition of $a$. Therefore,
\begin{align}\label{eq:minimax_tv_distance}
    \max_{\omega\in\Omega}\max_{j\in[M]}
    d_{\mathrm{TV}}(P_\omega,P_{\omega^{(j)}})
    \leq
    2\Phi\lt(\frac12\rt)-1.
\end{align}

Having established that $g_\omega\in\cH_d(\beta,L,H)$ and that neighboring
signals are difficult to distinguish, we complete the proof by reducing the
policy regret to a weighted Hamming risk over the hypercube.

For each $j\in[M]$, define $I_j:=\lt\{k\in[K]:x_k\in B(z_j,r/2)\rt\}$ and $v_j = \sum_{k \in I_j} p_k$. Note that $\{I_j\}_{j=1}^M$ are pairwise disjoint. For every
$k\in I_j$,
\[
    \phi_{z_j,r}(x_k)
    =
    \lt(1-\frac{\|x_k-z_j\|_2^2}{r^2}\rt)^\beta
    \geq \lt(1 - \frac{1}{4}\rt)^\beta = 
    \lt(\frac{3}{4}\rt)^\beta,
\]
and thus
\[
    |g_\omega(x_k)|
    \geq
    a\lt(\frac34\rt)^\beta.
\]
Moreover, by Assumption~\ref{ass:quasi_uniform_design}, 
\begin{align}\label{eq:positive_information}
    v_j = \sum_{k \in I_j} p_k = Q_K\lt(B(z_j,r/2)\rt) \geq c\lt( \frac{r}{2}\rt)^d > 0.
\end{align}

Fix a decision rule $\widetilde\pi:\R^K\to[0,1]^K$. Using that, for
$X\in[0,1]$ and $b\in\{0,1\}$,
$|\E[X]-b|=\E[|X-b|]$,
\begin{align*}
    \sup_{\substack{
        P_{g,\bm{\sigma}}\in\cP_K^0\\
        g\in\cH_d(\beta,L,H)
    }}
    R_K(P_{g,\bm{\sigma}},\widetilde\pi)
    &\geq
    \max_{\omega\in\Omega}
    R_K(P_{g_\omega,\bm{\overline{\sigma}}},\widetilde\pi)
    \\
    &=
    \max_{\omega\in\Omega}
    \sum_{k=1}^Kp_k|g_\omega(x_k)|
    \lt|\E_\omega[\widetilde\pi_k(Z)]
    -\indicator{g_\omega(x_k)>0}\rt|
    \\
    &=
    \max_{\omega\in\Omega}
    \sum_{k=1}^Kp_k|g_\omega(x_k)|
    \E_\omega\lt[
    |\widetilde\pi_k(Z)-\indicator{g_\omega(x_k)>0}|
    \rt]
    \\
    &\geq
    \max_{\omega\in\Omega}
    \sum_{j=1}^M\sum_{k\in I_j}
    p_k|g_\omega(x_k)|
    \E_\omega\lt[
    |\widetilde\pi_k(Z)-\indicator{g_\omega(x_k)>0}|
    \rt]
    \\
    &\geq
    a\lt(\frac34\rt)^\beta
    \max_{\omega\in\Omega}
    \sum_{j=1}^M\sum_{k\in I_j}
    p_k\E_\omega\lt[
    |\widetilde\pi_k(Z)-\indicator{g_\omega(x_k)>0}|
    \rt]
    \\
    &=
    a\lt(\frac34\rt)^\beta
    \max_{\omega\in\Omega}
    \sum_{j=1}^M\sum_{k\in I_j}
    p_k\E_\omega\lt[|\widetilde\pi_k(Z)-\omega_j|\rt].
\end{align*}
The final equality follows because, for every $k\in I_j$,
$\phi_{z_j,r}(x_k)>0$ and hence
$\indicator{g_\omega(x_k)>0}=\omega_j$.

Given $\widetilde\pi$, define the estimator
$\widetilde\omega:\R^K\to[0,1]^M$ entrywise by
\[
    \widetilde\omega_j(Z)
    :=
    \frac{1}{v_j}\sum_{k\in I_j}p_k\widetilde\pi_k(Z),
\]
where we recall $v_j > 0$ by Assumption~\ref{ass:quasi_uniform_design} and Equation~\eqref{eq:positive_information}. Since
$\omega_j\in\{0,1\}$ and
$\widetilde\pi_k(Z),\widetilde\omega_j(Z)\in[0,1]$,
\[
    \sum_{k\in I_j}p_k
    |\widetilde\pi_k(Z)-\omega_j|
    =
    v_j|\widetilde\omega_j(Z)-\omega_j|.
\]
Applying Lemma~\ref{lem:weighted_assouad} together with
Equation~\eqref{eq:minimax_tv_distance},
\begin{align*}
    \sup_{\substack{
        P_{g,\bm{\sigma}}\in\cP_K^0\\
        g\in\cH_d(\beta,L,H)
    }}
    R_K(P_{g,\bm{\sigma}},\widetilde\pi)
    &\geq
    a\lt(\frac34\rt)^\beta
    \max_{\omega\in\Omega}
    \sum_{j=1}^M\sum_{k\in I_j}
    p_k\E_\omega\lt[|\widetilde\pi_k(Z)-\omega_j|\rt]
    \\
    &=
    a\lt(\frac34\rt)^\beta
    \max_{\omega\in\Omega}
    \sum_{j=1}^M v_j
    \E_\omega\lt[|\widetilde\omega_j(Z)-\omega_j|\rt]
    \\
    &\geq
    a\lt(\frac34\rt)^\beta
    \inf_{\widetilde\omega}
    \max_{\omega\in\Omega}
    \sum_{j=1}^M v_j
    \E_\omega\lt[|\widetilde\omega_j(Z)-\omega_j|\rt]
    \\
    &\geq
    a\lt(\frac34\rt)^\beta
    \frac{1-\lt(2\Phi\lt(\frac12\rt)-1\rt)}{2}
    \sum_{j=1}^M v_j
    \\
    &=
    a\lt(\frac34\rt)^\beta
    \Phi\lt(-\frac12\rt)
    \sum_{j=1}^M v_j
    \\
    &=
    a\lt(\frac34\rt)^\beta
    \Phi\lt(-\frac12\rt)
    \sum_{j=1}^M Q_K\lt(B(z_j,r/2)\rt)
    \\
    &\geq
    a\lt(\frac34\rt)^\beta
    \Phi\lt(-\frac12\rt)
     \frac{c}{C2^{2d+1}}
\end{align*}
where the final inequality follows from the mass property in
Lemma~\ref{lem:low_norm_packing}. Since this holds for every measurable decision rule $\widetilde\pi:\R^K\to[0,1]^K$, taking the infimum and applying the reduction on $\cP_K^0 \subseteq \cP_K$ gives
\begin{align*}
    \inf_{\widetilde\pi:\R^K \times [0, \infty)^K \to[0,1]^K}
     \sup_{\substack{
        P_{g,\bm{\sigma}}\in\cP_K\\
        g\in\cH_d(\beta,L,H)
    }}
    R_K(P_{g,\bm{\sigma}},\widetilde\pi) 
    &\geq  \inf_{\widetilde\pi:\R^K \times [0, \infty)^K \to[0,1]^K} \sup_{\substack{
        P_{g,\bm{\sigma}}\in\cP_K^0\\
        g\in\cH_d(\beta,L,H)
    }}
    R_K(P_{g,\bm{\sigma}},\widetilde\pi)
    \\
    &=  \inf_{\widetilde\pi:\R^K \to[0,1]^K} \sup_{\substack{
        P_{g,\bm{\sigma}}\in\cP_K^0\\
        g\in\cH_d(\beta,L,H)
    }}
    R_K(P_{g,\bm{\sigma}},\widetilde\pi)
    \\
    &\geq
    a\lt(\frac34\rt)^\beta
    \Phi\lt(-\frac12\rt)
    \frac{c}{C2^{2d+1}}.
\end{align*}
By the definition of $a$,
\[
    a
    \geq
    \min\lt\{
        1,
        \frac{1}{2^{\beta+1}},
        \frac{1}{2\sqrt{2^{d+1}C}}
    \rt\}
    \min\lt\{
        H,
        Lr^\beta,
        \frac{\overline\sigma}{\sqrt{Kr^d}}
    \rt\}.
\]
Consequently,
\begin{align*}
     \inf_{\widetilde\pi:\R^K \times [0, \infty)^K \to[0,1]^K}
     \sup_{\substack{
        P_{g,\bm{\sigma}}\in\cP_K\\
        g\in\cH_d(\beta,L,H)
    }}
    R_K(P_{g,\bm{\sigma}},\widetilde\pi)  \geq
    c_1
    \min\lt\{
        H,
        Lr^\beta,
        \frac{\overline\sigma}{\sqrt{Kr^d}}
    \rt\},
\end{align*}
where
\[
    c_1
    :=
    \lt(\frac34\rt)^\beta
    \Phi\lt(-\frac12\rt)
    \frac{c}{C2^{2d+1}}
    \min\lt\{
        1,
        \frac{1}{2^{\beta+1}},
        \frac{1}{2\sqrt{2^{d+1}C}}
    \rt\}.
\]
The constant $c_1>0$ is independent of $K$ and $r$, proving the claim.

\subsection{Proof of Proposition~\ref{pro:welfare_decomposition}}\label{app:welfare_decomposition}
By Assumption~\ref{ass:gaussian_model}--\ref{ass:weak_signal}, $\widehat g_w(x_k) \sim \cN(m_w(x_k), s_w(x_k)^2)$. Therefore
\begin{align*}
    \sqrt{n} W_{P_{g, \bm{\sigma}}}(\widetilde\pi_w) &= \sum_{k = 1}^K p_k g(x_k) \E_{g, \bm{\sigma}}\lt[\widetilde\pi_{w,k}(Z)\rt]
    \\ 
    &= \sum_{k = 1}^K p_k g(x_k)  \E_{g, \bm{\sigma}}\lt[\indicator{\widehat g_w(x_k) > 0}\rt]
    \\ 
    &= \sum_{k = 1}^K p_k g(x_k)  P_{g, \bm{\sigma}}\lt(\widehat g_w(x_k) > 0\rt)
    \\ 
    &= \sum_{k = 1}^K p_k g(x_k) \Phi\lt(\frac{m_w(x_k)}{s_w(x_k)}\rt)
    \\
    &=  \sum_{k = 1}^K p_k g(x_k) \psi_{w, k}.
\end{align*}
Similarly,
\begin{align*}
    \Cov_p(g, \psi_w) &= \sum_{k = 1}^K p_k (g(x_k) - \bar g_p) (\psi_{w, k} - \bar \psi_{w, p}) 
    \\
    &=  \sum_{k = 1}^K p_k g(x_k) \psi_{w, k} - \bar g_p \bar \psi_{w, p}
\end{align*}
Combining the expressions gives
\begin{align*}
    \sqrt{n} W_{P_{g, \bm{\sigma}}}(\widetilde\pi_w) = \bar g_p \bar \psi_{w, p} +  \Cov_p(g, \psi_w),
\end{align*}
which proves the decomposition. Next, we show the pairwise representation of $\Cov_p(g, \psi_w)$.
\begin{align*}
    \Cov_p(g, \psi_w) &= \sum_{k = 1}^K p_k g(x_k) \psi_{w, k} - \bar g_p \bar \psi_{w, p}
    \\
    &= \sum_{k = 1}^K p_k g(x_k) \psi_{w, k} - \lt(\sum_{k=1}^K p_k g(x_k) \rt) \lt(\sum_{l=1}^K p_l \psi_{w, l} \rt)
    \\
    &= \sum_{k = 1}^K p_k g(x_k) \psi_{w, k} - \sum_{k=1}^K \sum_{l=1}^K p_k p_l g(x_k) \psi_{w, l}
    \\
    &= \frac{1}{2} \sum_{k = 1}^K p_k g(x_k) \psi_{w, k} + \frac{1}{2} \sum_{l = 1}^K p_l g(x_l) \psi_{w, l} - \frac{1}{2 }\sum_{k=1}^K \sum_{l=1}^K p_k p_l g(x_k) \psi_{w, l} - \frac{1}{2} \sum_{k=1}^K \sum_{l=1}^K p_k p_l g(x_l) \psi_{w, k}
    \\
    &= \frac{1}{2} \sum_{k = 1}^K p_k p_l \lt(g(x_k) \psi_{w, k} + g(x_l) \psi_{w, l} - g(x_k) \psi_{w, l} - g(x_l) \psi_{w, k}\rt)
    \\
    &= \frac{1}{2} \sum_{k = 1}^K \sum_{j = 1}^K p_k p_j (g(x_k) - g(x_j)) (\psi_{w, k} - \psi_{w, l}).
\end{align*}

\subsection{Proof of Proposition~\ref{pro:cross_validation}}\label{app:cross_validation}

Define the scaled welfare loss $L_{\bm h}(\pi) := \sqrt{n} \{V_{\bm h}^* - V_{\bm h}(\pi)\} = \sum_{k =1}^K p_k h_k (\indicator{h_k > 0} - \pi_k)$. For each $r\in\Lambda_K$, write 
\[
    \bar L_{\bm h,r} := \frac{1}{F} \sum_{f = 1}^F L_{\bm h}(\widetilde\pi_{w_r}(Z^{(-f)})) = \frac{1}{F} \sum_{f = 1}^F \sum_{k =1}^K p_k h_k (\indicator{h_k > 0} - \widetilde\pi_{w_r,k}(Z^{(-f)})).
\]
Then $\E_{P_{\bm h, \bm \sigma}}[\bar L_{\bm h,r}] = R_K(P_{\bm h, \bm \sigma_F}, \widetilde\pi_{w_r})$. 

Next, for each fold $f$, define
\[
    \Delta_{r, f} := \sum_{k = 1}^K p_k (Z_k^{(f)} - h_k) \widetilde \pi_{w_r, k}(Z^{(-f)}).
\]
Since $Z_k^{(f)}=h_k+(Z_k^{(f)}-h_k)$, our selection criterion (depending only on observables) decomposes
\[
    \widehat W(r) := \frac{1}{F} \sum_{f = 1}^F \sum_{k = 1}^K p_k Z_k^{(f)} \widetilde\pi_{w_r,k}(Z^{(-f)}) = \sqrt{n} V_{\bm h}^* - \bar L_{\bm h,r} + \frac{1}{F} \sum_{f = 1}^F \Delta_{r, f} .
\]
Thus,  if $\widehat r$ maximizes $\widehat W(r)$ it also maximizes $- \bar L_{\bm h,r} + \frac{1}{F} \sum_{f = 1}^F \Delta_{r, f}$. Thus for every fixed $r \in \Lambda_K$, 
\begin{align}\label{eq:fold_error}
    \bar L_{\bm h, \widehat r} &\leq \bar L_{\bm h,r} + \frac{1}{F} \sum_{f = 1}^F ( \Delta_{\widehat r, f} -  \Delta_{r, f})
    \notag
    \\
    &\leq \bar L_{\bm h,r} + \frac{1}{F} \sum_{f = 1}^F \lt( \max_{u \in \Lambda_K} \Delta_{u, f} -  \Delta_{r, f}\rt)
\end{align}

Now, conditional on the remaining folds $\cD_{-f} = (Z^{(j)})_{j \neq f}$, $\Delta_{r, f}$ is a centered Gaussian random variable with variance bounded by $s_K^2 = F \overline{\sigma}^2 \sum_{k = 1}^K p_k^2 \leq \frac{C F \overline{\sigma}^2}{K}$, since $\max_k p_k \lesssim K^{-1}$ by assumption. By the Gaussian maximum inequality \citet[Exercise 2.12(a)]{wainwright2019high},
\[
    \E_{P_{\bm h, \bm \sigma}}[\max_{r \in \Lambda_K} \Delta_{r, f} \mid \cD_{-f}] \leq s_K \sqrt{2 \log(|\Lambda_K|)}.
\]
Thus, taking the expectation in Equation~\eqref{eq:fold_error} and minimizing over $r$ gives 
\[
    \E_{P_{\bm h, \bm \sigma}}[\bar L_{\bm h, \widehat r}] \leq \min_{r \in \Lambda_K} \E_{P_{\bm h, \bm \sigma}}[\bar L_{\bm h,r}] + s_K \sqrt{2 \log(|\Lambda_K|)}
\]
Assume for now that $R_K(P_{\bm h,\bm\sigma},\widetilde\pi_{w_{\widehat r}}) \leq F\E_{P_{\bm h,\bm\sigma}}[\bar L_{\bm h,\widehat r}]$. Then, since $\E_{P_{\bm h, \bm \sigma}}[\bar L_{\bm h,r}] = R_K(P_{\bm h, \bm \sigma_F}, \widetilde\pi_{w_r})$ and $s_K \leq \sqrt{\frac{C F \overline{\sigma}^2}{K}}$ by the above, 
\begin{align*}
        R_K(P_{\bm h,\bm\sigma},\widetilde\pi_{w_{\widehat r}}) &\leq F\E_{P_{\bm h,\bm\sigma}}[\bar L_{\bm h,\widehat r}]
       \\
       &\leq F \min_{r \in \Lambda_K} \E_{P_{\bm h, \bm \sigma}}[\bar L_{\bm h, r}] + F s_K \sqrt{2 \log(|\Lambda_K|)}
        \\
        &\leq F \min_{r \in \Lambda_K} R_K(P_{\bm h, \bm \sigma_F}, \widetilde\pi_{w_r}) + F \sqrt{\frac{C F \overline{\sigma}^2 2 \log(|\Lambda_K|)}{K}}
        \\
        &\lesssim \min_{r \in \Lambda_K} R_K(P_{\bm h, \bm \sigma_F}, \widetilde\pi_{w_r}) + \sqrt{\frac{\log(|\Lambda_K|)}{K}},
\end{align*}
which proves the claim as $F$ is fixed. It remains to show that
\[
    R_K(P_{\bm h,\bm\sigma},\widetilde\pi_{w_{\widehat r}}) \leq F\E_{P_{\bm h,\bm\sigma}} [\bar L_{\bm h,\widehat r}].
\]

Since $Z_k = \frac{1}{F} \sum_{f = 1}^F Z_k^{(f)}$, we have
\begin{align*}
    &L_{\bm h}(\widetilde \pi_{w_r}(Z)) 
    \\
    &= \sqrt{n} \{V_{\bm h}^* - V_{\bm h}(\widetilde \pi_{w_r}(Z))\}
    \\
    &= \sum_{k =1}^K p_k h_k (\indicator{h_k > 0} - \widetilde \pi_{w_r, k}(Z))
     \\
    &= \sum_{k =1}^K p_k h_k \lt(\indicator{h_k > 0} - \indicator{\sum_{j = 1}^K w_{r, j}(x_k) Z_j > 0}\rt)
    \\
    &= \sum_{k =1}^K p_k h_k \lt(\indicator{h_k > 0} - \indicator{\frac{1}{F} \sum_{f = 1}^F \sum_{j = 1}^K w_{r, j}(x_k) Z_j^{(-f)} > 0}\rt)
    \\
    &= \sum_{k:h_k>0} p_k h_k \lt(1 - \indicator{\frac{1}{F}\sum_{f=1}^F\sum_{j=1}^K w_{r,j}(x_k)Z_j^{(-f)} > 0}\rt) + \sum_{k:h_k<0} p_k(-h_k) \indicator{\frac{1}{F}\sum_{f=1}^F\sum_{j=1}^K w_{r,j}(x_k)Z_j^{(-f)} > 0}
    \\
    &= \sum_{k:h_k>0} p_k h_k \indicator{\frac{1}{F}\sum_{f=1}^F\sum_{j=1}^K w_{r,j}(x_k)Z_j^{(-f)} \leq 0} + \sum_{k:h_k<0} p_k(-h_k) \indicator{\frac{1}{F}\sum_{f=1}^F\sum_{j=1}^K w_{r,j}(x_k)Z_j^{(-f)} > 0}
    \\   
    &\leq \sum_{k:h_k>0} p_k h_k \sum_{f=1}^F \indicator{\sum_{j=1}^K w_{r,j}(x_k)Z_j^{(-f)} \leq 0} + \sum_{k:h_k<0} p_k(-h_k) \sum_{f=1}^F \indicator{\sum_{j=1}^K w_{r,j}(x_k)Z_j^{(-f)} > 0}
    \\
    &= \sum_{f=1}^F\sum_{k=1}^K p_k h_k \lt(\indicator{h_k>0} - \indicator{\sum_{j=1}^K w_{r,j}(x_k)Z_j^{(-f)} > 0}\rt)
    \\
    &= \sum_{f=1}^F L_{\bm h}(\widetilde\pi_{w_r}(Z^{(-f)}))
    \\
    &= F \bar L_{\bm h,r}.
\end{align*}
This inequality holds for every $r\in\Lambda_K$, including $\widehat r$. Setting $r = \widehat r$ and taking expectations gives
\[
    R_K(P_{\bm h,\bm\sigma},\widetilde\pi_{w_{\widehat r}}) = \E_{P_{\bm h,\bm\sigma}}[L_{\bm h}(\widetilde\pi_{w_{\widehat r}}(Z))] \leq F\E_{P_{\bm h,\bm\sigma}}[\bar L_{\bm h,\widehat r}],
\]
which completes the proof.

\newpage

\subsection{Technical Lemmas}

\begin{lemma}[Weighted Assouad Lemma for Randomized Decisions]
\label{lem:weighted_assouad}
Let $M,K\geq1$, let $\Omega:=\{0,1\}^M$, and let
$\{P_\omega:\omega\in\Omega\}$ be Borel probability measures on
$\R^K$.
For $j\in[M]$, let $\omega^{(j)}$ denote the vector
obtained from $\omega$ by flipping its $j$th coordinate, i.e., $\omega^{(j)} := (\omega_1,\ldots,\omega_{j-1}, 1-\omega_j,\omega_{j+1},\ldots,\omega_M).$
Suppose that
\[
    \max_{\omega\in\Omega}
    \max_{j\in[M]}
    \TV\!\lt(P_\omega,P_{\omega^{(j)}}\rt)
    \leq \Delta,
\]
for some $\Delta\in[0,1)$, where $\TV(P,Q) := \sup_{A\in\mathcal B(\R^K)} |P(A)-Q(A)|$ is the total variation distance. Then, for any nonnegative weights $v_1,\ldots,v_M$,
\[
    \inf_{\widetilde\omega:\R^K\to[0,1]^M}
    \max_{\omega\in\Omega}
    \E_\omega\!\lt[
        \sum_{j=1}^M
        v_j
        |\widetilde\omega_j(Z)-\omega_j|
    \rt]
    \geq
    \frac{1-\Delta}{2}
    \sum_{j=1}^M v_j,
\]
where the infimum ranges over all measurable decision rules
$\widetilde\omega:\R^K\to[0,1]^M$, and $\E_\omega$ denotes expectation under $P_\omega$. The same lower bound holds with the maximum replaced by the average over $\Omega$.
\end{lemma}

\begin{proof}
The proof follows from similar arguments as in
\citet[Theorem~2.12]{tsybakov2009introduction}, but extended to
nonnegative weights and randomized decisions.

Fix any measurable decision rule
$\widetilde\omega:\R^K\to[0,1]^M$. Bounding the maximum by the average over $\Omega$ gives
\begin{align*}
    \max_{\omega\in\Omega} \E_\omega \lt[ \sum_{j=1}^M v_j|\widetilde\omega_j(Z)-\omega_j| \rt]
    &\geq \frac{1}{2^M} \sum_{\omega\in\Omega}
   \E_\omega \lt[ \sum_{j=1}^M v_j|\widetilde\omega_j(Z)-\omega_j| \rt]
   \\
    &=
    \frac{1}{2^M}
    \sum_{j=1}^M v_j
    \sum_{\omega\in\Omega}
    \E_\omega \lt[ |\widetilde\omega_j(Z)-\omega_j| \rt].
\end{align*}

For each $j\in[M]$, define $\Omega_{j,0} := \{\omega\in\Omega:\omega_j=0\}$. Using $\widetilde\omega_j(Z)\in[0,1]$,
\begin{align*}
    \max_{\omega\in\Omega} \E_\omega \lt[ \sum_{j=1}^M v_j|\widetilde\omega_j(Z)-\omega_j| \rt]
    & \geq
    \frac{1}{2^M}
    \sum_{j=1}^M v_j
    \sum_{\omega\in\Omega}
    \E_\omega \lt[ |\widetilde\omega_j(Z)-\omega_j| \rt]
    \\
    &= \frac{1}{2^M}
    \sum_{j=1}^M v_j
    \sum_{\omega\in\Omega_{j,0}}
    \Big\{
        \E_\omega[|\widetilde\omega_j(Z)|]
        +
        \E_{\omega^{(j)}}[|1-\widetilde\omega_j(Z)|]
    \Big\}
    \\
    &= \frac{1}{2^M}
    \sum_{j=1}^M v_j
    \sum_{\omega\in\Omega_{j,0}}
    \Big\{
        \E_\omega[\widetilde\omega_j(Z)]
        +
        \E_{\omega^{(j)}}[1-\widetilde\omega_j(Z)]
    \Big\}
    \\
    &= \frac{1}{2^M}
    \sum_{j=1}^M v_j
    \sum_{\omega\in\Omega_{j,0}}
    \lt[
        1 -
        \Big\{\E_{\omega^{(j)}}[\widetilde\omega_j(Z)] - \E_\omega[\widetilde\omega_j(Z)]
    \Big\} \rt]
    \\
    &\geq \frac{1}{2^M}
    \sum_{j=1}^M v_j
    \sum_{\omega\in\Omega_{j,0}}
    \lt[
        1 -
        \lt|\E_{\omega^{(j)}}[\widetilde\omega_j(Z)] - \E_\omega[\widetilde\omega_j(Z)]
    \rt| \rt]
\end{align*}
For any measurable $\delta:\R^K\to[0,1]$ and any probability measures $P,Q$ on $\R^K$, 
\begin{align*}
    \lt| \E_P[\delta(Z)]-\E_Q[\delta(Z)] \rt|
    &=\lt| \int_0^1 \lt\{ P(\delta(Z)>t)-Q(\delta(Z)>t) \rt\} dt \rt|
    \\
    &\leq \int_0^1 \lt| P(\delta(Z)>t)-Q(\delta(Z)>t) \rt| dt    
    \\
    &\leq \int_0^1 \TV(P,Q) dt
    \\
    &= \TV(P,Q).
\end{align*}
Applying this inequality with $\delta=\widetilde\omega_j$, $P=P_{\omega^{(j)}}$, and $Q=P_\omega$ and recalling that $\TV(P_\omega, P_{\omega^{(j)}}) \leq \Delta$ for all $\omega \in \Omega, j \in [M]$,
\begin{align*}
     \max_{\omega\in\Omega} \E_\omega \lt[ \sum_{j=1}^M v_j|\widetilde\omega_j(Z)-\omega_j| \rt]
    & \geq \frac{1}{2^M}
    \sum_{j=1}^M v_j
    \sum_{\omega\in\Omega_{j,0}}
    \lt[
        1 -
        \lt|\E_{\omega^{(j)}}[\widetilde\omega_j(Z)] - \E_\omega[\widetilde\omega_j(Z)]
    \rt| \rt]
    \\
    &\geq \frac{1}{2^M}
    \sum_{j=1}^M v_j
    \sum_{\omega\in\Omega_{j,0}}
    \lt[ 1 - \TV(P_\omega, P_{\omega^{(j)}}) \rt]
    \\
    &\geq \frac{1}{2^M}
    \sum_{j=1}^M v_j
    \sum_{\omega\in\Omega_{j,0}} \lt[ 1 - \Delta \rt]
    \\
    &= \frac{1 - \Delta}{2} \sum_{j=1}^M v_j,
\end{align*}
where we used $|\Omega_{j,0}|=2^{M-1}$ for every $j\in[M]$. Since this holds for every measurable
$\widetilde\omega:\R^K\to[0,1]^M$, taking the infimum over
$\widetilde\omega$ proves the claim.
\end{proof}

\begin{lemma}[Low-Norm Packing]
\label{lem:low_norm_packing}
Suppose Assumption~\ref{ass:quasi_uniform_design} holds with constants $0 < c < C < \infty, a_0, r_0 > 0$ and let $r\in[2a_0K^{-1/d},r_0/2]$. Define $\phi_0(x):=(1-\|x\|_2^2)_+^\beta$ and set $\phi_{z,r}(x):=\phi_0\lt(\frac{x-z}{r}\rt)$. Then there exists $\cZ_r=\{z_1,\ldots,z_M\}  \subseteq \{x_k:k\in[K]\}$
satisfying the following properties:
\begin{enumerate}[label=(\alph*)]
    \item \emph{Separation:} $\|z_i-z_j\|_2>2r$ for $i \neq j \in [M]$.
    \item \emph{Cardinality:} $ M\geq\frac{1}{C2^{d+1}r^d}.$
    \item \emph{Task Norm:} For every $j\in[M]$, $\sum_{k=1}^K\phi_{z_j,r}(x_k)^2 \leq 2^{d+1}CKr^d$.
    \item \emph{Mass:} $\sum_{j=1}^M Q_K\lt(B(z_j,r/2)\rt) \geq \frac{c}{C2^{2d+1}}$.
\end{enumerate}
\end{lemma}

\begin{proof}
Let
$\widetilde{\cZ}_r
=\{\widetilde z_1,\ldots,\widetilde z_{\widetilde M}\}$
be a maximal $2r$-packing of the distinct task locations, i.e.,
$\|\widetilde z_i-\widetilde z_j\|_2>2r$ for $i\neq j$ and
\[
    \{x_k:k\in[K]\}
    \subseteq
    \bigcup_{j=1}^{\widetilde M}B(\widetilde z_j,2r).
\]
Such a set exists because there are finitely many distinct task locations.
Hence, by Assumption~\ref{ass:quasi_uniform_design},
\[
    1
    =
    Q_K\bigl(\{x_k:k\in[K]\}\bigr)
    \leq
    \sum_{j=1}^{\widetilde M}
    Q_K\bigl(B(\widetilde z_j,2r)\bigr)
    \leq
    \widetilde M C(2r)^d,
\]
and therefore $ \widetilde M \geq \frac{1}{C2^dr^d}$.

Next, since $\operatorname{supp}(\phi_0)=B(0,1)$, we have
$\operatorname{supp}(\phi_{z,r})=B(z,r)$. Thus, because
$\widetilde{\cZ}_r$ is a $2r$-packing, the supports of
$\{\phi_{\widetilde z_j,r}:j\in[\widetilde M]\}$ are pairwise disjoint.
Consequently, since $0\leq\phi_{z,r}\leq1$,
\[
    \sum_{j=1}^{\widetilde M} \sum_{k=1}^K
    \phi_{\widetilde z_j,r}(x_k)^2
    \leq K.
\]
Combining this with the lower bound on $\widetilde M$ gives
\[
    \frac{1}{\widetilde M}
    \sum_{j=1}^{\widetilde M}
    \sum_{k=1}^K \phi_{\widetilde z_j,r}(x_k)^2
    \leq
    \frac{K}{\widetilde M}
    \leq
    C2^dKr^d.
\]
Therefore, at least half $\widetilde z_j \in \widetilde{\cZ}_r$ satisfies
\[
    \sum_{k=1}^K \phi_{\widetilde z_j,r}(x_k)^2 \leq
    2^{d+1}CKr^d;
\]
otherwise, the preceding average would exceed $C2^dKr^d$.

Only retain these elements and relabel them as
$\cZ_r=\{z_1,\ldots,z_M\}$. Then separation still holds and
\[
    M \geq \frac{\widetilde M}{2} \geq \frac{1}{C2^{d+1}r^d},
\]
by construction.
Moreover, for every $j\in[M]$,
\[
    \sum_{k=1}^K\phi_{z_j,r}(x_k)^2
    \leq
    2^{d+1}CKr^d.
\]

Finally, the balls $\{B(z_j,r/2):j\in[M]\}$ are pairwise disjoint.
Hence, by Assumption~\ref{ass:quasi_uniform_design},
\[
    \sum_{j=1}^M Q_K\lt(B(z_j,r/2)\rt)
    \geq
    Mc(r/2)^d
    \geq
    \frac{c}{C2^{2d+1}}.
\]
\end{proof}

\begin{lemma}[H\"older Regularity]
\label{lem:hoelder_regularity}
Fix $d\geq1$ and $0<\beta\leq1$. Let $x, z \in \R^d, r > 0$. Define $\phi_0(x):=(1-\|x\|_2^2)_+^\beta$ and set $\phi_{z,r}(x) := \phi_0\lt(\frac{x-z}{r}\rt)$. Then, for every $x,y\in\R^d$,
\[
    |\phi_{z,r}(x)-\phi_{z,r}(y)| \leq \frac{2^\beta}{r^\beta}\|x-y\|_2^\beta.
\]
\end{lemma}

\begin{proof}
    Define $f(x):=(1-\|x\|_2^2)_+$ such that $\phi_{z,r}(x) = h\lt(\frac{x - z}{r}\rt)^\beta$. Let $x, y \in \R^d$. We show that $f$ is $2$-Lipschitz.
    
    First, assume $\|x\|_2,\|y\|_2\leq1$. Then by the reverse triangle inequality 
    \[
        |f(x)-f(y)| = \lt|\|x\|_2^2-\|y\|_2^2\rt| \leq 2\|x-y\|_2.
    \]
    Similarly, if $\|x\|_2,\|y\|_2 > 1$ then $f(x) = f(y) = 0$ and the condition holds trivially. Assume now that, without loss of
    generality, $\|x\|_2\leq1<\|y\|_2$. Then by the reverse triangle inequality again
    \[
        |f(x)-f(y)| = f(x) = (1+\|x\|_2) (1-\|x\|_2) \leq 2 (1-\|x\|_2)
        \leq 2(\|y\|_2-\|x\|_2) \leq 2\|x-y\|_2.
    \]
    Thus $f$ is globally $2$-Lipschitz. 
    
    Since $0<\beta\leq1$,
    $|u^\beta-v^\beta|\leq|u-v|^\beta$ for all $u,v\geq0$. Consequently,
    \[
        |\phi_0(x)-\phi_0(y)|
        =
        |f(x)^\beta-f(y)^\beta|
        \leq
        2^\beta\|x-y\|_2^\beta,
    \]
    which shows that $\phi_0$ is $\beta$-H\"older with constant $2^\beta$. Finally, 
    \[
         |\phi_{z,r}(x)-\phi_{z,r}(y)| = \lt|\phi_0\lt(\frac{x-z}{r}\rt) - \phi_0\lt(\frac{y-z}{r}\rt)\rt| \leq 2^\beta \lt\| \frac{x - z}{r} - \frac{y - z}{r} \rt\|_2^\beta \leq \frac{2^\beta}{r^\beta}\|x-y\|_2^\beta.
    \]
\end{proof}

\end{document}